\RequirePackage{fix-cm}

\makeatletter
\providecommand{\input@path}{} 
\g@addto@macro{\input@path}{{template/}} 
\makeatother
\documentclass[twocolumn]{svjour3}          

\smartqed  
\usepackage{graphicx} 
\usepackage{xcolor} 
\usepackage{subcaption}
\usepackage{multirow}
\usepackage{algorithm}
\usepackage{makecell}  
\usepackage{siunitx}
\usepackage{stfloats}
\usepackage{booktabs}  
\usepackage{url}  
\usepackage{hyperref} 
\usepackage{amsmath} 
\usepackage[noEnd=true,indLines=true]{algpseudocodex}

\usepackage{mathptmx}      
\journalname{arXiv}
\begin{document}

\title{Disk-Based Interval Indexes Under the Increasing Ending Time Assumption
}

\author{Kai Wang         \and
	Moin Hussian Moti \and
	Dimitris Papadias
}

\institute{K. Wang \at
	Department of Computer Science and Engineering, HKUST, Clearwater Bay, Hong Kong \\
	\email{kwangbn@connect.ust.hk}
	\and
	M. Hussian Moti \at
	Department of Computer Science and Engineering, HKUST, Clearwater Bay, Hong Kong \\
	\email{mhmoti@connect.ust.hk}
	\and
	D. Papadias \at
	Department of Computer Science and Engineering, HKUST, Clearwater Bay, Hong Kong \\
	\email{dimitris@cs.ust.hk}
}

\date{Received: date / Accepted: date}

\maketitle

\begin{abstract}
	Indexes for large collections of intervals are common in temporal databases, where each record has a lifespan, or validity interval. Despite their conceptual differences, we demonstrate that interval indexes can be captured by some corner structure in a 2D space. This representation facilitates the optimization of query processing by identifying nodes that must contain query results versus nodes that may contain results. In addition, we explore the assumption that intervals arrive in order of increasing ending time (IET) to develop disk-based indexes that have compact size, efficient insertions, and fast query processing. Specifically, we first develop CEB, an index in the corner space defined by the interval center and endpoint. Our second contribution is TIDE, which organizes intervals by their duration and endpoint. CEB and TIDE adopt a two-layer architecture, where the leaf nodes of a top tree (ordering intervals by their center or duration) correspond to the root nodes of bottom trees, ordering intervals by their endpoints. Both top and bottom trees are append-only B+-trees to facilitate fast insertions. CEB and TIDE outperform state-of-the-art competitors in terms of index size and insertion speed. In addition, TIDE is always faster in query processing, sometimes by orders of magnitude. 
\end{abstract}

\section{Introduction}
In temporal databases each record has a lifespan, or validity interval $[t_s, t_e)$, where $t_s<t_e$. A record is considered alive if $t_e$ equals the current time $(now)$; otherwise, it is dead. The most important queries on intervals are \emph{stabbing} and \emph{range} queries, which retrieve all intervals intersecting a timestamp or period in history. These query types form the basic building blocks of more complex tasks. Their efficient processing necessitates interval indexes, which depending on the application requirements, can be disk or main-memory based. They can also be classified as static, if they assume that all data intervals are known in advance, or dynamic if they allow insertions of new intervals. Dynamic indexes are \textit{partially persistent} \cite{Driscoll1986Persis,Salzberg1999survey}, if updates can only occur at the current time. Consequently, nodes storing dead intervals are \emph{immutable}. On the other hand, \textit{fully persistent} \cite{Lanka1991FullyPersis,Brodal2023FullyPersis} structures allow updates at any point in the past.

We demonstrate that interval indexes can be captured by some corner structure in a 2D space, defined by selecting two out of four possible dimensions including starting time $t_s$, ending time $t_e$, duration $d$, and center point $c$. This unified representation facilitates the optimization of query processing by identifying nodes that \emph{must} contain query results (i.e., all their intervals can be directly reported) versus nodes that \emph{may} contain results (i.e., their intervals must be examined individually).
Moreover, the representation provides useful insight into the advantages and shortcomings of each index. Specifically, some structures have highly unbalanced nodes, while others involve redundancy, which necessitates duplicate elimination during query processing.

In this paper, we focus on disk-based interval indexes, assuming that intervals arrive in increasing order of their ending time $t_e$, i.e., they are added to the database when they die. In addition to being realistic \footnote{IET can be utilized even for applications where the intervals arrive in a different order, e.g., $t_s$. In this case alive intervals can be maintained by a main-memory index on $t_s$ and, when they die, they are moved to the disk-based index on $t_e$. In practice, since alive intervals are usually a small fraction of the dead ones, it is reasonable to maintain them in memory.}, this Increasing Ending Time (IET) assumption yields a form of partial persistence that has significant advantages on the index properties: (i) it leads to the existence of immutable nodes that are full (or can be compressed), thus minimizing the total index size; (ii) it enables append-only insertions, which are very fast. The only other index aimed at IET is SEB \cite{Song2003SEB}, which is based on a corner structure defined by $t_s$ and $t_e$. SEB includes a two-level scheme: a top B+-tree indexes $t_s$, while intervals with similar values of $t_s$ are grouped together in bottom B+-trees ordered by $t_e$. New intervals can only be appended to the last node of the bottom tree that covers the corresponding $t_s$, and the rest of the nodes are immutable.

The main problem of SEB is the existence of numerous (thousands or millions depending on the dataset) bottom trees, each with a mutable last node, which affects negatively insertion and query performance. To mitigate this issue, we first propose \textbf{CEB} (\textbf{C}enter and \textbf{E}ndpoint \textbf{B}-tree), which aims at improving SEB by decreasing the number of mutable nodes. CEB also consists of two layers of append-only B+-trees. The top-tree is ordered by the center point ($c=\frac{t_s+t_e}{2}$) of intervals.  The leaf nodes of the top tree correspond to the root nodes of bottom trees, which are B+-trees ordered by $t_e$. The advantage of CEB compared to SEB is that the bottom trees with maximum $c$ below $\frac{now}{2}$ will never receive insertions, and can be compressed to immutable full nodes.

Although CEB indeed reduces the bottom trees compared to SEB, their number continuously grows with time because $c$, similar to  $t_s$, is ever increasing. The only interval property (among $t_s$, $t_e$, $c$, and $d$) that does not necessarily grow with time is the duration $d$.   This observation led to our second contribution \textbf{TIDE} (\textbf{T}ime \textbf{I}ntervals by \textbf{D}uration and \textbf{E}ndpoint), where the top tree organizes intervals by $d$, and the bottom trees by $t_e$.
This yields a very small number of bottom trees, especially for datasets that involve low duration variance. Although in TIDE, the last node of each bottom tree is mutable, the total number of mutable nodes is much lower than SEB and CEB, facilitating high compactness and cache locality.

We evaluated CEB and TIDE against SEB \cite{Song2003SEB} and the RI-tree \cite{Kriegel2000RItree}, the state-of-the-art disk-based interval index, using real datasets with diverse characteristics. CEB and SEB have, in general, similar overall cost, and both outperform the RI-tree on insertion performance and index size. On the other hand, the RI-tree is faster on query processing. TIDE is always the fastest on every aspect, often outperforming the rest by orders of magnitude. Moreover, it achieves better space efficiency because its insertion strategy generates full nodes.
The rest of the paper is organized as follows. Section \ref{sec:related} reviews existing interval indexes. Section \ref{sec:unified} describes a representation that enables optimization of query processing and conceptual evaluation of different indexes under a unifying framework. Section \ref{sec:methodCEB} and Section \ref{sec:methodTIDE} present the insertion and query processing algorithms of CEB and TIDE, respectively. Section \ref{sec:exp} compares TIDE and CEB against SEB and the RI-tree under different metrics and datasets, and Section \ref{sec:conclude} concludes the paper.

\section{Related work}\label{sec:related}
Section \ref{subsec:relatedInmemory} focuses on main-memory, and Section \ref{subsec:relatedDisk} on disk-based interval indexes.

\subsection{Main-memory interval indexes} \label{subsec:relatedInmemory}

In the \textit{interval-tree} \cite{Edelsbrunner1980ItvTree}, intervals are first sorted on their endpoints, and are then organized in a binary search tree BST. Figure \ref{subfig:intervalTree} shows the BST for 13 intervals $s_0,...,s_{12}$, in increasing ending time. The top node $n_{14}$ corresponds to median endpoint at time 14. An interval is stored at the highest overlapping node of BST; e.g., $s_8$, $s_{11}$ are stored at $n_{25}$. In each BST node, the stored intervals are duplicated in two sorted lists: ascending order of $t_s$, and descending order of $t_e$. Assume for instance, a range query [10,13] in Figure \ref{subfig:intervalTree}. All intervals assigned to BST nodes $n_{12}$, $n_{13}$ within [10,13] qualify the query. In addition, some nodes on either side of the range ($n_6$, $n_9$, $n_{14}$) may contain results.
Specifically, for BST nodes $n_6$, $n_9$ before [10,13], we scan the $t_e$ list (descending order) and report all intervals, until reaching the first one with $t_e<10$. For BST node $n_{14}$ after [10,13], we scan the $t_s$ list (ascending order) and report all intervals, until reaching the first one with $t_s>13$.
No other BST nodes, may contribute results.

\begin{figure}[t]
	\centering
	\begin{subfigure}{0.52\linewidth}
		\includegraphics[width=\linewidth]{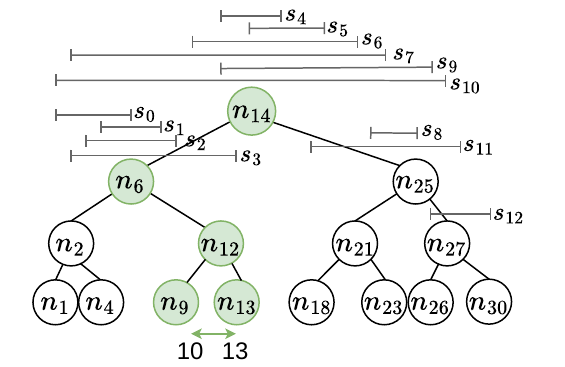}
		\caption{interval-tree}
		\label{subfig:intervalTree}
	\end{subfigure}%
	\begin{subfigure}{0.48\linewidth}
		\includegraphics[width=\linewidth]{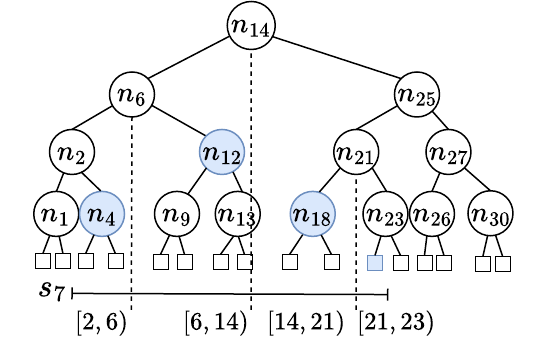}
		\caption{segment-tree}
		\label{subfig:segmentTree}
	\end{subfigure}%
	\caption{Classical structures}
	\label{fig:base}
\end{figure}

Interval-tree nodes may be very unbalanced. For example, in Figure \ref{subfig:intervalTree}, the root node $n_{14}$ contains numerous intervals ($s_4$, $s_5$, $s_6$, $s_7$, $s_9$, $s_{10}$), while several BST nodes, e.g., $n_1$, $n_2$, $n_{12}$, are empty.
To avoid this issue, the \textit{segment-tree} \cite{Berg2000SegTree} partitions and stores each interval in multiple nodes. Specifically, an interval is first stored at all leaf nodes that intersect it. Then, consecutive partitions are merged at the upper level recursively, if they fully cover their parent node. For example, in Figure \ref{subfig:segmentTree}, interval $s_7=[2,23)$ is stored in $n_4$, $n_{12}$, $n_{18}$ and $n_{23}^L$,
where $L$ ($R$) denotes left (right) leaf node. Due to replication, the segment-tree consumes $O(N\log N)$ space, where $N$ is the interval cardinality. Range query processing requires duplicate elimination since the same interval may exist in several nodes, possibly at different levels.

\begin{figure}[t]
	\centering
	\includegraphics[width=0.8\linewidth]{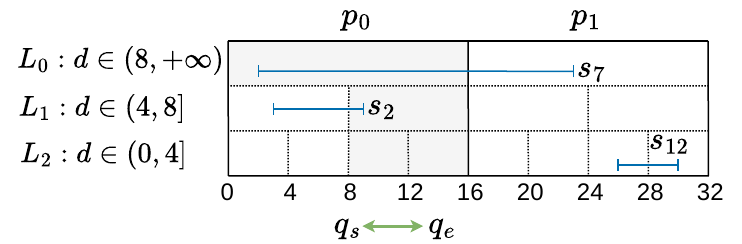}
	\caption{Period-index}
	\label{fig:PeriodIndex}
\end{figure}

A learned \textit{period-index} \cite{Behrend2019PeriodIdx} splits the time domain into temporal periods (e.g., days) and divides each period hierarchically.
Figure \ref{fig:PeriodIndex} shows thirteen intervals in two periods $p_0$ and $p_1$.
Each period has length $l=16$ and is partitioned into $h=3$ levels. Both $l$ and $h$ are learned parameters.
An interval is stored at the top level such that its duration is more than half of the extent of that level.
For example, $s_2=[3, 9)$ is assigned to $L_1$ because its length (6) exceeds half the length of level 1.
Since $s_2$ intersects two partitions of $L_1$, it is stored in both.
A range query searches all levels intersecting its range. Compared to the segment-tree, the period-index incurs less redundancy (because all copies of an interval are at the same level), but duplicate elimination is still necessary for range queries.

\begin{figure}[t]
	\centering
	\includegraphics[width=0.75\linewidth]{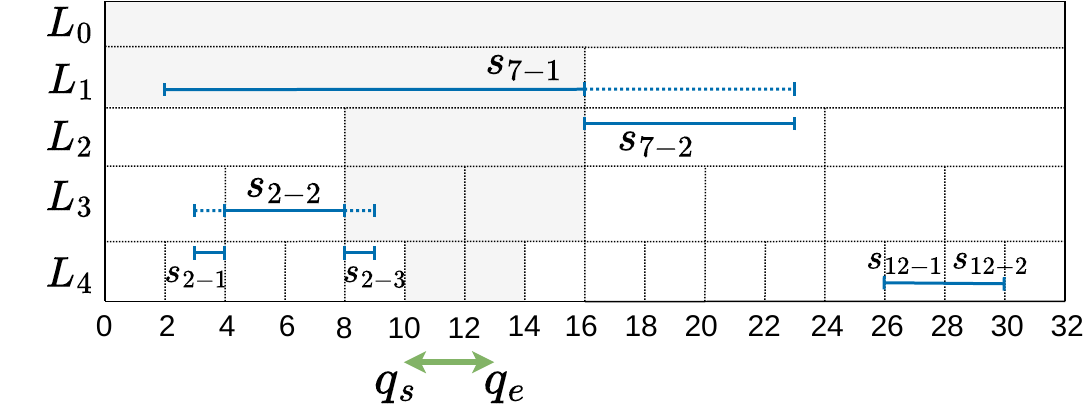}
	\caption{HINT}
	\label{fig:HINT}
\end{figure}

\textit{HINT} \cite{George2022HINT} also applies a hierarchical decomposition of buckets into $h$ levels. In Figure \ref{fig:HINT}, there are $h=5$ levels: the top $L_0$ covers the entire domain, while each bucket in $L_4$ covers two timestamps. An interval is first stored at all bottom level partitions that intersect it. Then, consecutive partitions with the same parent are merged at the upper level recursively, if they partially intersect (as opposed to be fully covered as in the segment-tree).  For instance, $s_7=[2,23)$ is stored with only two duplicates (i.e., $s_{7-1}$ in level $L_1$ and $s_{7-2}$ in level $L_2$), while the segment-tree requires four (see Figure \ref{subfig:segmentTree}). In other cases (e.g., $s_2$, $s_{12}$), both indexes have the same number of duplicates.
The first occurrence of an interval in HINT is marked as \textit{original}, and the rest are \textit{replicas}.
During query processing, for all partitions intersecting the start of the query range, all data intervals are examined.
For the remaining partitions, only originals may constitute results (replicas correspond to duplicates).
\textit{LIT} \cite{George2024LIT} extends HINT for dynamic intervals.
The \textit{RD-tree} \cite{Ceccarello2023RDtree} has two variants:
RD-tree-td, sorts intervals by $t_s$ and duration $d$; RD-tree-dt, sorts intervals by $d$ and $t_s$.
Finally, interval indexes have been applied for specialized tasks including temporal aggregation \cite{Mamoulis2010CornerR}, interval sampling \cite{Amagata2024Sampling} and interval joins \cite{Bouros2025HINTJoin,Bouros2021optFS,Hu2022TempJ}.

\begin{table}[!thb]
	\centering
	\small
		\begin{tabular}{c|c|c|c}
			\toprule
			Method                                       & Domain  & Partitioning & \#Copies/Interval     \\
			\midrule
			interval-tree \cite{Edelsbrunner1980ItvTree} & fixed   & data-driven  & 2                     \\
			\cmidrule(lr){1-4}
			segment-tree \cite{Berg2000SegTree}          & fixed   & data-driven  & $\leq 2\cdot log_2 N$ \\
			\cmidrule(lr){1-4}
			period-index \cite{Behrend2019PeriodIdx}     & growing & space-driven & $\leq max(2, n_b)$    \\
			\cmidrule(lr){1-4}
			HINT \cite{George2022HINT}                   & fixed   & space-driven & $\leq 2\cdot h$       \\
			\cmidrule(lr){1-4}
			RD-tree \cite{Ceccarello2023RDtree}          & fixed   & data-driven  & 1                     \\
            \midrule
            \multicolumn{4}{p{0.95\linewidth}}{\footnotesize
            \parbox{\linewidth}{
            \textit{*} $n_b$ is the number of buckets in the period-index.\\
            \textit{*} $h$ is the number of levels in HINT.
            }}
		\end{tabular}
	\caption{Summary of in-memory interval indexes}
	\label{table:memoryRelate}
\end{table}

Table \ref{table:memoryRelate} summarizes the properties of in-memory interval indexes. Most indexes assume static intervals over a fixed domain, and are not suitable for ever-evolving time. Though the period-index assumes a growing temporal domain, it requires a fixed period length, determined in advance. Indexes can also be classified as \emph{space-driven} (or \emph{data-driven}), if they partition the temporal domain using regular ranges/grids (or based on the data intervals). Finally, most in-memory interval indexes incur redundancy (i.e., they store each interval more than once), requiring some form of duplicate elimination during query processing.

\subsection{Disk-resident interval indexes} \label{subsec:relatedDisk}

The \textit{external interval-tree} (\textit{EI-tree}) and \textit{external segment-tree} (\textit{ES-tree}) \cite{Arge1996external} are disk-based extensions of their main-memory counterparts that replace the binary search tree (BST) with a B+-tree of fanout $\sqrt{B}$, where $B$ is the disk page capacity.
Similar to the in-memory structures, they maintain two sorted lists (with unlimited capacity) per B+-tree node, which may lead to expensive updates. For example, a split would force numerous intervals to move between nodes. To decrease the update cost, the EI-tree and ES-tree use complicated \textit{buffer tree} and \textit{weight balancing} techniques, which are theoretical in nature.

The \textit{time-index} \cite{Elmasri1990TimeIdx} stores alive records at the first timestamp of each node, and incremental updates at the following timestamps.
The \textit{append only-tree} (\textit{AP-tree}) \cite{Segev1993APtree} orders intervals by their starting time $t_s$ using an append-only B+-tree.
The AP-tree has optimal insertion speed, but is slow to answer stabbing queries.
Another trend extends R-tree and its variants \cite{Sellis1987RPlus,Beckmann1990RStar} to manage intervals.
However, R-trees are not effective for long intervals and high overlaps \cite{Kriegel2000RItree}.
To deal with long intervals, the \textit{segment R-tree} (\textit{SR-tree}) \cite{Kolovson1991SegR} combines the main-memory segment-tree with the disk-based R-tree.
Similar to the segment-tree, intervals in the SR-tree are stored in both leaf and internal nodes, leading to redundancy.
Therefore, it consumes $O(\frac{N}{B}\log_B N)$ space \cite{Salzberg1999survey}.

\begin{figure}[t]
	\centering
	\begin{subfigure}{0.38\linewidth}
		\includegraphics[width=\linewidth]{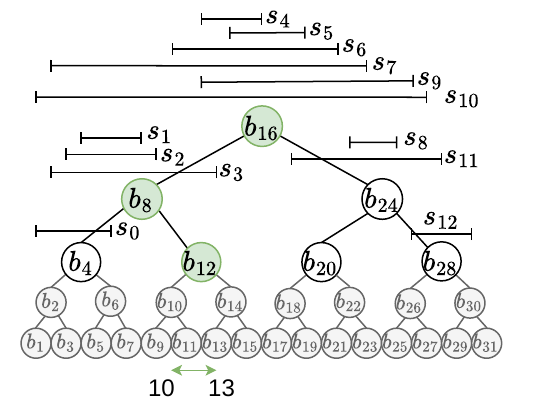}
		\caption{Decomposition}
		\label{subfig:RITbst}
	\end{subfigure}%
	\begin{subfigure}{0.62\linewidth}
		\includegraphics[width=\linewidth]{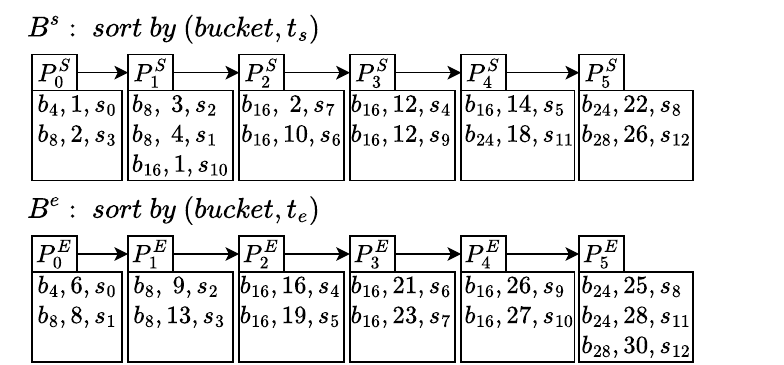}
		\caption{Leaf pages of two B+-trees}
		\label{subfig:RITbtree}
	\end{subfigure}%
	\caption{Relational interval-tree}
	\label{fig:RIT}
\end{figure}

The \textit{relational interval-tree} (\textit{RI-tree}) \cite{Kriegel2000RItree} assumes the domain to be a range $(0,2^h)$, where $h$ denotes the minimum integer that fulfills $2^h > now$. Similar to HINT, it maintains a hierarchical decomposition with $h$ levels. The partitioning is space-driven, and remains valid until $now \geq 2^h$, in which case $h$ increases by 1 to expand the domain.
For instance, Figure \ref{subfig:RITbst} shows a decomposition with height $h=5$ and $now=30$, assuming the same intervals as Figure \ref{subfig:intervalTree}. Its root bucket $b_{16}$  corresponds to the middle of $2^h=32$ and leaf nodes (e.g., $b_1$, $b_3$) to individual timestamps.
As shown in Figure \ref{subfig:RITbtree}, the RI-tree maintains two B+-trees, $B^s$,$B^e$, ordering intervals by composite key $(bucket, t_s)$ and $(bucket, t_e)$, respectively.
An interval is assigned to the highest overlapped bucket and inserted into both B+-trees.
For instance, $s_0=[1,6)$ is assigned to $b_4$, and an entry $<b_4,1,s_0>$ ($<b_4,6,s_0>$) is inserted to $B^s$ ($B^e$).
Figure \ref{subfig:RITbtree} shows the leaf nodes of $B^s$,$B^e$ assuming that disk page capacity is 3. Observe that (i) each interval is stored once per tree, (ii) the occurrences of a bucket per tree equals the number of assigned intervals, and (iii) a bucket may be stored in multiple consecutive pages.

Given a range query of $[10,13]$, processing starts from the root bucket $b_{16}$. Since $b_{16}$ is after the range, the first entry for $b_{16}$ (i.e., $<b_{16},1,s_{10}>$ in Figure \ref{subfig:RITbtree}) is found at $B^s$; $s_{10}$ and all subsequent intervals (e.g., $s_7$, $s_6$, $s_4$, $s_9$) until the first with $t_s>13$ ($s_5$), are returned as results. At the next level, $b_{24}$ cannot contain results because all its intervals must start after $b_{16}$. On the other hand, $b_8$ may contain intervals that start before timestamp 8 and finish before 16. To retrieve these results, we locate $b_8$ at $B^e$ and output all its intervals that end after time 10 ($s_3$). \footnote{Each intersected bucket $b_i$ requires a range query, either from $<b_i,-\infty>$ to $<b_i,13>$ in $B^s$, or from $<b_i,10>$ to $<b_i,+\infty>$ in $B^e$.} Bucket $b_{12}$ could potentially contain results, but it is empty. The lowest two levels are not searched, because the RI-tree records the deepest level that receives insertions (and all buckets at these levels are empty).

\begin{figure}[t]
	\centering
	\begin{subfigure}{0.49\linewidth}
		\includegraphics[width=\linewidth]{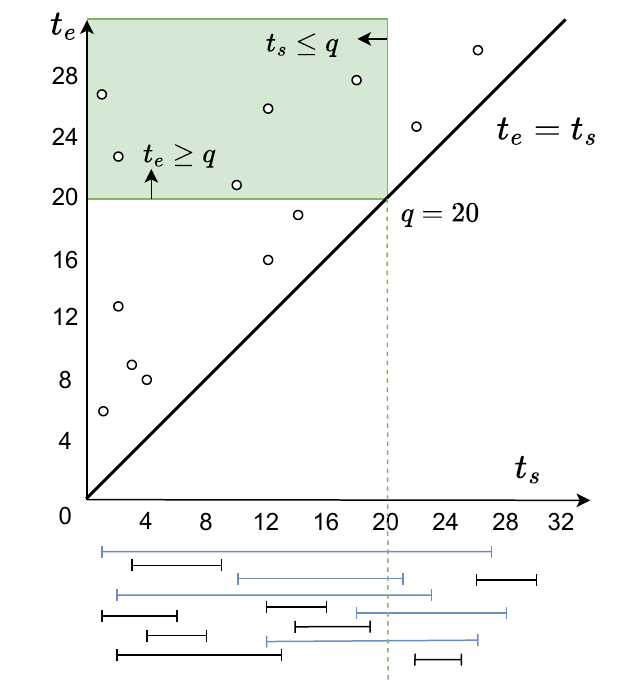}
		\caption{stabbing query}
		\label{subfig:CornerStab}
	\end{subfigure}%
	\begin{subfigure}{0.49\linewidth}
		\includegraphics[width=\linewidth]{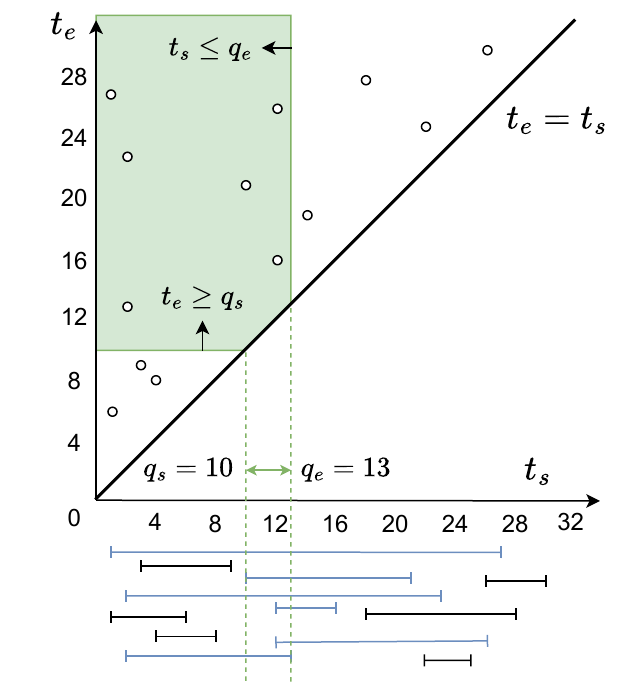}
		\caption{range query}
		\label{subfig:CornerRange}
	\end{subfigure}%
	\caption{Diagonal corner queries}
	\label{fig:CornerQueries}
\end{figure}

\cite{Kanellakis1993corner} maps intervals to a 2-dimensional (2D) space, where $t_s$ is the horizontal and $t_e$ is the vertical axis.
The mapped points lie above line $t_e=t_s$, because $t_e > t_s$, forming a \textit{diagonal corner structure} \footnote{We use the terms corner structure and corner space interchangeably.}.
For example, in Figure \ref{subfig:CornerStab}, a stabbing query at time 20 retrieves all points (intervals) in the shaded area.
In Figure \ref{subfig:CornerRange}, a query with range $[10, 13]$ returns points with $t_s \leq 13$ and $t_e \geq 10$.
Early work on corner structures has been of theoretical nature, assuming static data \cite{Salzberg1999survey,Vitter2001external}.
Space-partitioning Generalized search trees (SP-GiST\footnote{\url{https://github.com/postgres/postgres/blob/REL_10_STABLE/src/backend/utils/adt/rangetypes_spgist.c}}) \cite{Eltabakh2006SPGiST} maintain points in corner spaces using spatial indexes such as KD-trees or R-trees in PostgreSQL. Following a similar idea, TDSQL manages mapped points with R-tree variants \cite{Lu2019TDSQL}. DOT \cite{Faloutsos1991DOT} indexes points in corner spaces by a B+-tree after applying space-filling curves (e.g., Peano or Hilbert) to generate 1D order.
The above approaches assumed that mapped points are uniformly distributed, but it was later shown that corner spaces are usually skewed \cite{Gaede1998survey}, exhibiting higher density near the diagonal $t_e=t_s$.

\begin{figure}[!thb]
	\centering
	\begin{subfigure}{0.49\linewidth}
		\includegraphics[width=\linewidth]{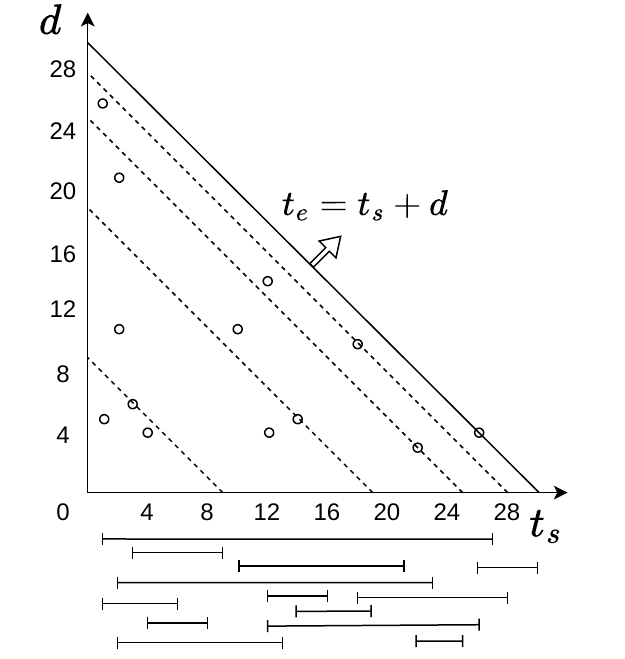}
		\caption{D-ordering by $t_e,t_s$}
		\label{subfig:ISTD}
	\end{subfigure}%
	\begin{subfigure}{0.49\linewidth}
		\includegraphics[width=\linewidth]{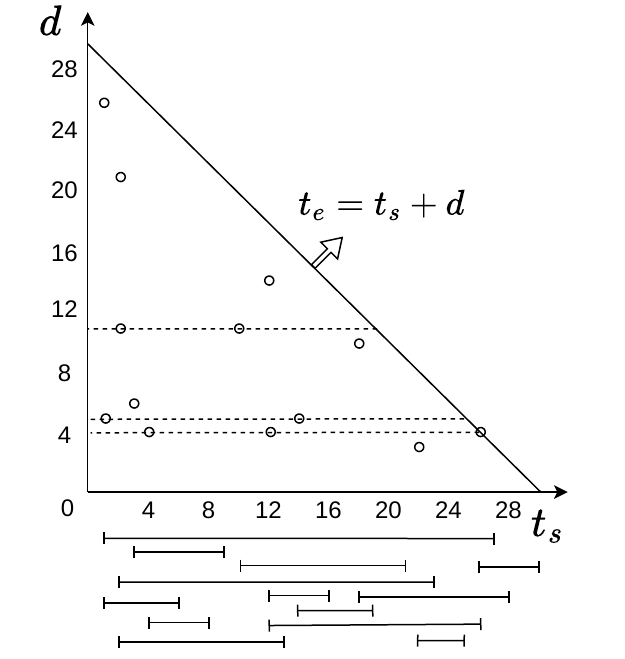}
		\caption{H-ordering by $d,t_s$}
		\label{subfig:ISTH}
	\end{subfigure}%
	\caption{Interval spatial transformation}
	\label{fig:IST}
\end{figure}

The \textit{interval spatial transformation} (\textit{IST}) \cite{Goh1996IST} maps intervals into a 2D space of $t_s$ (x-axis) and $d$ (y-axis), where $d$ is the duration. IST forms a triangle starting from point $(0,0)$ and bounded by line $t_e= t_s+d \leq now$.
The points are sorted and indexed by a single B+-tree. Depending on the ordering, there are three variants: D(iagonal) (sort by $t_e,t_s$), V(ertical) (sort by $t_s,t_e$), and H(orizontal) (sort by $d,t_s$).  Figure \ref{subfig:ISTD} shows D-ordering, assuming a B+-tree node capacity of 3. The first node contains the three intervals closest to $(0,0)$ with the smallest $t_e$. Figure \ref{subfig:ISTH} shows an example of H-ordering, where the first node of B+-tree contains the three intervals with the shortest duration.
Instead of endpoints, \emph{midpoint transformation} \cite{Nievergelt1987midpoint,Seeger1988midpoint} maps interval centers $\frac{t_s+t_e}{2}$ (x-axis) and expansions $\frac{t_e-t_s}{2}$ (y-axis). Figure \ref{subfig:MidS} shows the mapped space, which is bounded by $t_s=x-y\geq 0$ and $t_e=x+y\leq now$. Midpoint transformations are limited to theoretical interest due to potentially complicated queries \cite{Gaede1998survey}. For instance, as shown in Figure \ref{subfig:MidRQ}, a range query $[10, 13]$ returning points with $t_s \leq 13$ and $t_e \geq 10$ has a complex shape and boundaries.

\begin{figure}[!thb]
	\centering
	\begin{subfigure}{0.49\linewidth}
		\includegraphics[width=\linewidth]{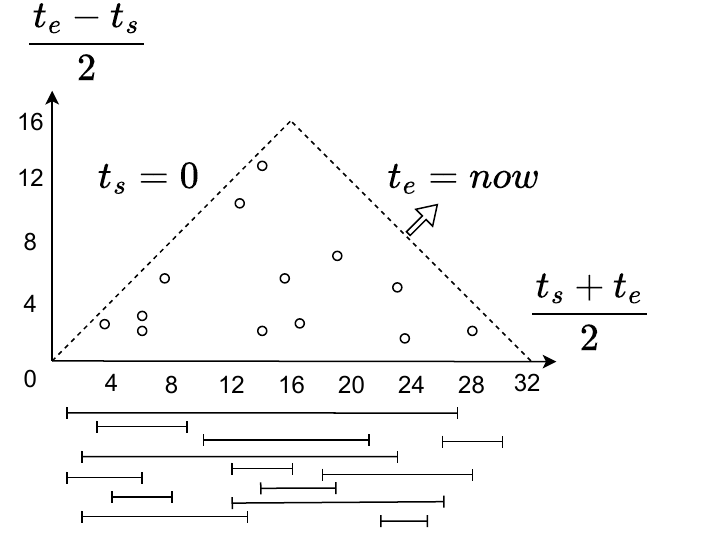}
		\caption{the structure}
		\label{subfig:MidS}
	\end{subfigure}%
	\begin{subfigure}{0.49\linewidth}
		\includegraphics[width=\linewidth]{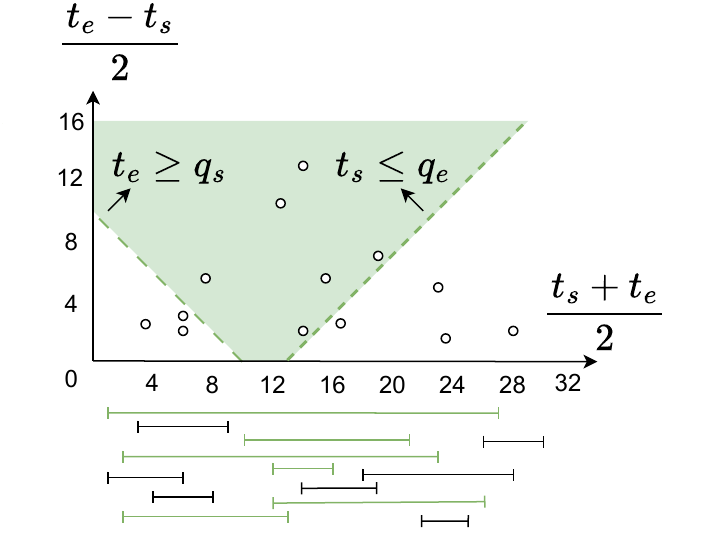}
		\caption{query with range [10, 13]}
		\label{subfig:MidRQ}
	\end{subfigure}%
	\caption{Midpoint transformation}
	\label{fig:Mid}
\end{figure}

The \textit{start/end timestamp B-tree} (SEB) \cite{Song2003SEB} applies corner structures for indexing intervals following the IET assumption: intervals are inserted in increasing order of $t_e$. Figure
\ref{subfig:SEBstruc} shows SEB for thirteen intervals in six data nodes with capacity 3.
Initially, there is a single data node $D_0$ that covers the entire diagonal corner space, $t_s \in [0, +\infty)$ and $t_e \in (0, +\infty)$.
When $D_0$ overflows at time 9, it generates two nodes $D_1$ and $D_2$.
$D_1$ is the \emph{first} node in its column with a new $t_s$ range $(4, +\infty)$.
$D_2$ is a node with the same $t_s$ range as $D_0$ $[0, 4]$, and a different $t_e$ range $[9, +\infty)$.
$D_0$ is full and becomes immutable.
$D_1$ and $D_2$ are non-full and ready to accept insertions.
A new point always falls into $D_1$ (if its $t_s > 4$) or $D_2$ (if $t_s \leq 4$).
When a \emph{non-first} node overflows, it only requires horizontal partitioning.
For example, at time 27, the overflow of $D_2$ creates node $D_5$ with the same $t_s$ range $[0,4]$ as $D_2$.
Points in $D_2$ fulfill $t_e \in [9,27]$, while points in $D_5$ have $t_e \in [27, +\infty)$.

\begin{figure}[t]
	\centering
	\begin{subfigure}{0.49\linewidth}
		\includegraphics[width=\linewidth]{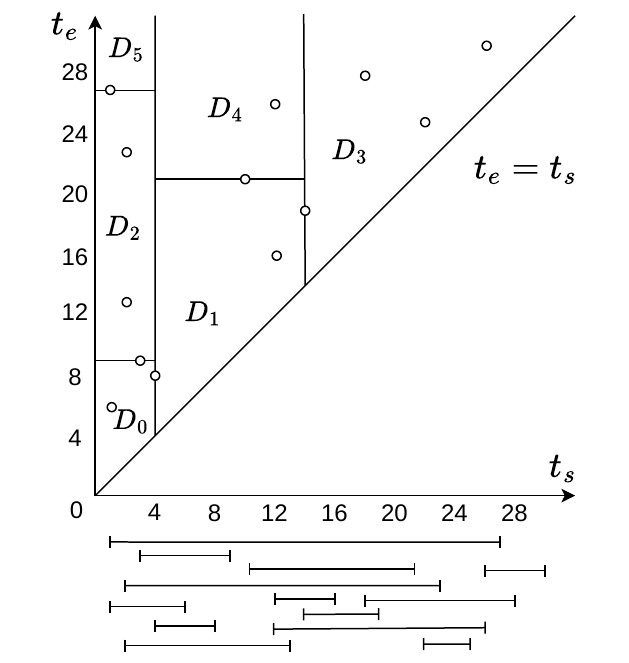}
		\caption{the structure}
		\label{subfig:SEBstruc}
	\end{subfigure}%
	\begin{subfigure}{0.49\linewidth}
		\includegraphics[width=\linewidth]{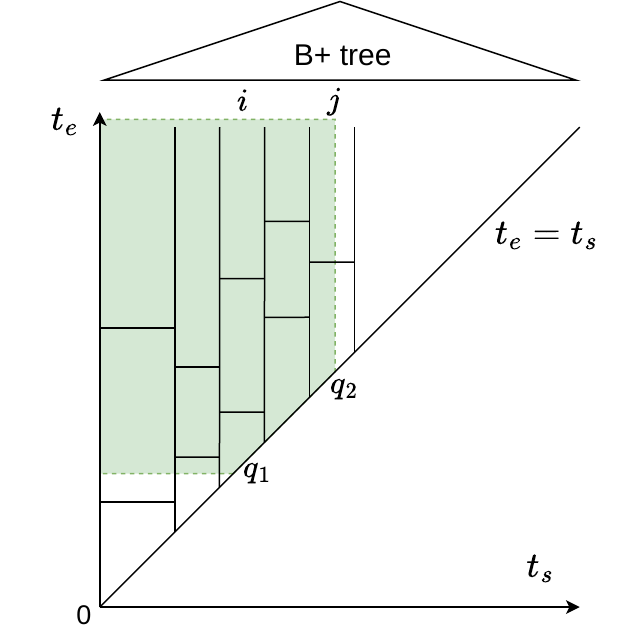}
		\caption{range query}
		\label{subfig:SEBquery}
	\end{subfigure}%
	\caption{SEB}
	\label{fig:SEB}
\end{figure}

SEB uses a top-layer B+-tree to index the columns, which correspond to nodes with the same $t_s$ range, as shown in Figure \ref{subfig:SEBquery}. Each column is indexed by a separate B+-tree, called a bottom tree (BT).
An insertion of new interval $[t_s,t_e)$ first locates the bottom tree that corresponds to $t_s$. Given the IET assumption, point $(t_s,t_e)$ is appended to the last node of that BT.
A range query $[q_s,q_e]$ identifies column $i$ (and $j$) covering $q_s$ (and $q_e$).
All nodes in columns up to $j$ are examined for results, provided that their $t_e$ exceeds $q_s$.
SEB has been used mostly for trajectory management \cite{Cudre2010Trajstore,Patel2004STRIPES,De2005MonTree}.
The \textit{compressed start end-tree} \cite{Wang2008CSE} applies similar concepts, but switches the order from $<t_s, t_e>$ to $<t_e, t_s>$. This however introduces high update cost because insertions may occur at any node (as opposed to the last node of some column as in SEB).

\begin{table*}[!thb]
	\centering
	\small
	\begin{tabular}{c|c|c|c|c}
		\toprule
		Method                            & Domain  & Partitioning & \#Copies/Interval & Disk-based Index                                              \\
		\midrule
		RI-tree \cite{Kriegel2000RItree}  & growing & space-driven & 2                 & two B+-trees ordered by $(bucket,t_s)$ and $(bucket, t_e)$    \\
		\cmidrule(lr){1-5}
		SP-GiST \cite{Eltabakh2006SPGiST} & fixed   & data-driven  & 1                 & KD-tree / R-tree                                              \\
		\cmidrule(lr){1-5}
		DOT \cite{Faloutsos1991DOT}       & fixed   & space-driven & 1                 & one B+-tree ordered by space-filling curves                   \\
		\cmidrule(lr){1-5}
		IST \cite{Goh1996IST}             & growing & data-driven  & 1                 & one B+-tree ordered by $(t_e,t_s)$  / $(t_s,t_e)$ / $(d,t_s)$ \\
		\cmidrule(lr){1-5}
		SEB \cite{Song2003SEB}            & growing & data-driven  & 1                 & two-level B+-trees ordered by $t_s$ and $t_e$                 \\
		\bottomrule
	\end{tabular}
	\caption{Disk-resident interval indexes based on corner structures}
	\label{table:mapRelate}
\end{table*}

Table \ref{table:mapRelate} summarizes the practical disk-resident interval indexes, including the corresponding structure in the last column. Most methods assume a growing time domain, excluding SP-GiST \cite{Eltabakh2006SPGiST} and DOT \cite{Faloutsos1991DOT}), which require nontrivial extensions. The RI-tree and DOT are space-driven, whereas the rest are data-driven. Except for the RI-tree that stores each interval twice, the rest have no redundancy. Only the SEB utilizes IET, allowing for efficient append-only insertions. In the RI-tree and IST insertions may happen at any node, which may lead to numerous nodes that are below capacity.

\section{Corner Structures for Interval Indexes} \label{sec:unified}

In this section, we demonstrate that interval indexes, in general, can be captured by some corner structure in a 2D space, defined by the endpoints $t_s$, $t_e$, duration $d$, or center point $c$. This representation enables the identification of nodes that \emph{must} contain query results (i.e., all their intervals can be directly reported) versus nodes that \emph{may} contain results (i.e., their intervals must be individually examined). In addition to reducing the computation cost of regular queries, this may also decrease the I/O cost of \emph{aggregate} queries. For instance, when we wish to compute the \emph{count} of intervals intersecting a range, the number of intervals within each node inside the range can be aggregated directly, without visiting the node. This is particularly beneficial for large ranges that contain multiple nodes, possibly at high levels.

\begin{figure}[t]
	\centering
	\begin{subfigure}{0.49\linewidth}
		\includegraphics[width=\linewidth]{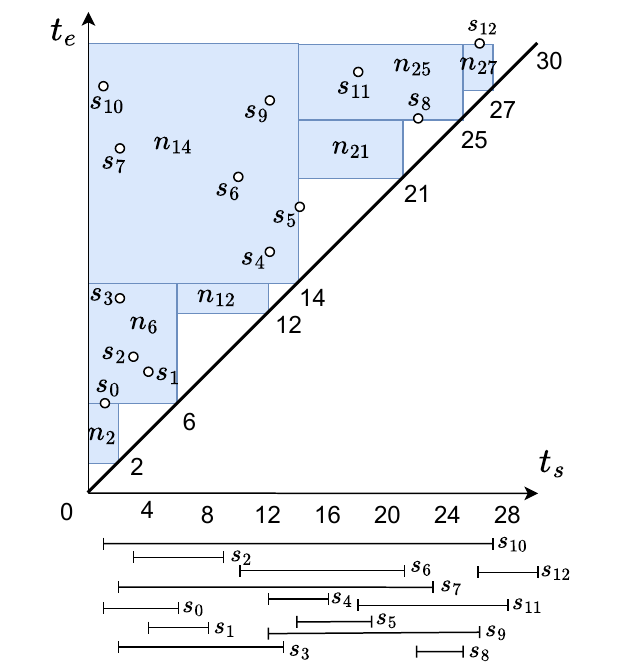}
		\caption{the structure}
		\label{subfig:ItvCornerStruc}
	\end{subfigure}%
	\begin{subfigure}{0.49\linewidth}
		\includegraphics[width=\linewidth]{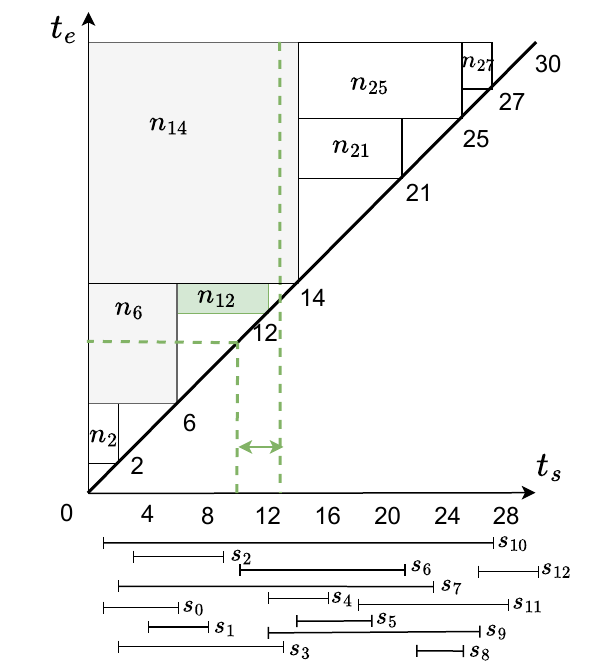}
		\caption{query with range $[10,13]$}
		\label{subfig:ItvCornerQuery}
	\end{subfigure}%
	\caption{Corner structure of the interval-tree}
	\label{fig:ItvCorner}
\end{figure}

First, we focus on the interval-tree. Figure \ref{subfig:ItvCornerStruc} maps the nodes and 13 intervals of Figure  \ref{subfig:intervalTree} into a corner structure, where $t_s$ is the $x$-axis and $t_e$ the $y$-axis. Each node corresponds to a rectangular area in the mapped space.
For example, node $n_{14}$ stores intervals intersecting with time 14, i.e., its mapped area is $t_s \leq 14 \leq t_e$. Similarly, $n_6$ is mapped to the space of $t_s \leq 6 \leq t_e < 14$, and $n_2$ corresponds to the space of $t_s \leq 2 \leq t_e < 6$. Figure \ref{subfig:ItvCornerQuery} shows the processing of a range query $[10, 13]$. Points in nodes ($n_{6}$, $n_{14}$) partially intersecting the range must be examined because they may constitute results (for these intervals $t_s \leq 13$ and $t_e \geq 10$). On the other hand, all intervals in nodes covered by the range ($n_{12}$) are directly reported, because they are query results (for these intervals $t_s \geq 10$ and $t_e \leq 13$).

\begin{figure}[t]
	\centering
	\begin{subfigure}{0.49\linewidth}
		\includegraphics[width=\linewidth]{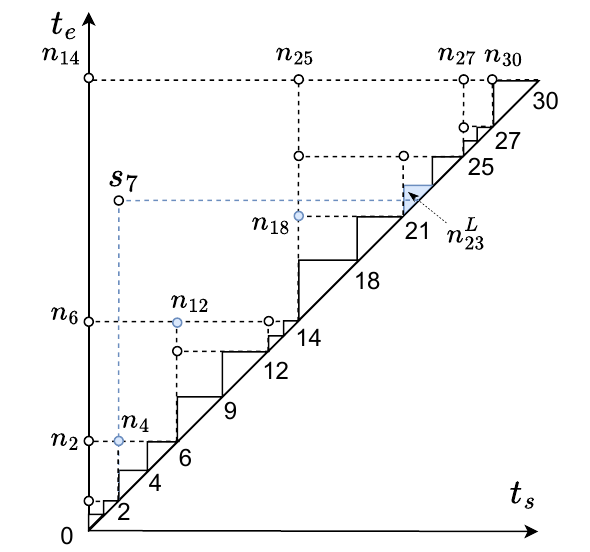}
		\caption{the structure}
		\label{subfig:SegCornerStruc}
	\end{subfigure}%
	\begin{subfigure}{0.49\linewidth}
		\includegraphics[width=\linewidth]{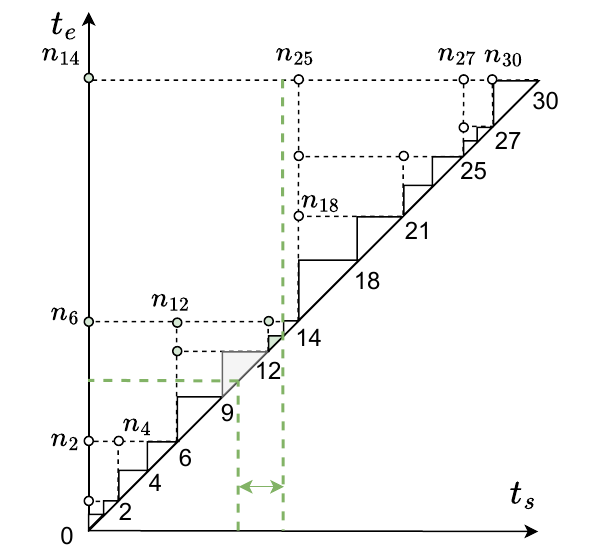}
		\caption{query with range $[10,13]$}
		\label{subfig:SegCornerQuery}
	\end{subfigure}%
	\caption{Corner structure of the segment-tree}
	\label{fig:SegCorner}
\end{figure}

Similar to the interval-tree, the corner structure for the segment-tree organizes intervals in nodes based on $t_s$ and $t_e$, without explicitly considering the duration $d$. Figure \ref{subfig:SegCornerStruc} maps the nodes of Figure \ref{subfig:segmentTree} into a 2D space that has a bottom layer of triangle-shaped leaf nodes. Recall that an interval may be partitioned and stored in multiple nodes, possibly at different levels.
For instance, $s_7=[2,23)$ is first assigned to leaf nodes, which are merged recursively, if the parent node is fully covered by $s_7$. The last partition containing $t_e$ of $s_7$ is stored in leaf node $n_{23}^L$. The remaining partitions are merged in internal nodes. Specifically, copies of $s_7$ are stored in nodes $n_4=[2,6)$, $n_{12}=[6,14)$, $n_{18}=[14,21)$ and $n_{23}^L=[21,23)$. Internal nodes are mapped to points because they only store intervals covering their full range. A query with range $[10, 13]$ in Figure \ref{subfig:SegCornerQuery} reports directly the intervals of green-shaded nodes (i.e., $n_{6}$, $n_{9}$, $n_{12}$, $n_{13}^L$, $n_{13}$, $n_{14}$) in the area defined by $t_s \geq 10$ and $t_e \leq 13$. Intervals in grey nodes (i.e., $n_{9}^R$) require inspection.
Duplicate elimination is necessary.

\begin{figure}[t]
	\centering
	\begin{subfigure}{0.49\linewidth}
		\includegraphics[width=\linewidth]{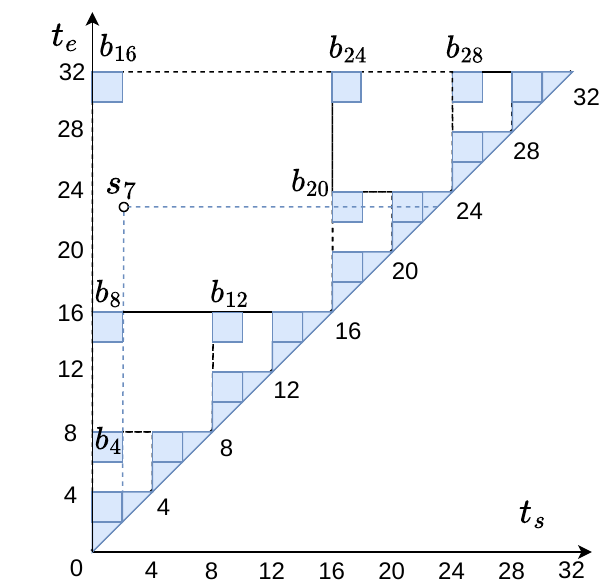}
		\caption{the structure}
		\label{subfig:HINTCornerStruc}
	\end{subfigure}%
	\begin{subfigure}{0.49\linewidth}
		\includegraphics[width=\linewidth]{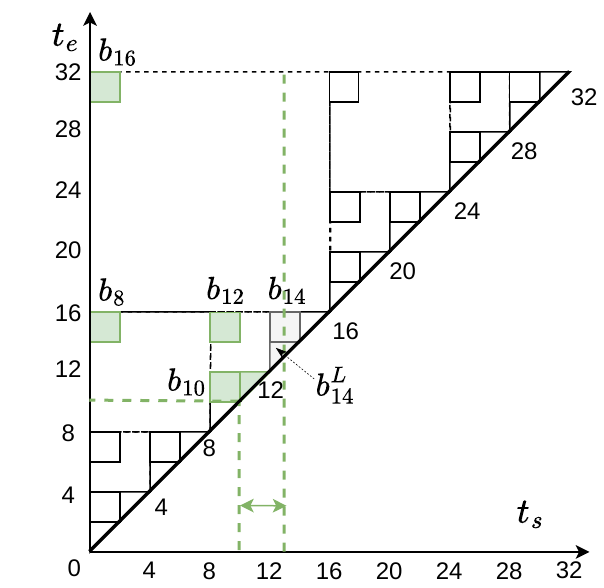}
		\caption{query with range $[10,13]$}
		\label{subfig:HINTCornerQuery}
	\end{subfigure}%
	\caption{Corner structure of HINT}
	\label{fig:HINTCorner}
\end{figure}

The corner structure of HINT \cite{George2022HINT} resembles that of the segment-tree, except that internal buckets correspond to square-shaped areas (instead of points). Both indexes assign intervals to leaf nodes/buckets (represented as triangles) first, and then merge upwards. HINT merges consecutive buckets with the same parent, if they \emph{partially} intersect an interval (as opposed to fully intersect in the segment-tree). Each internal bucket is represented as a square in Figure \ref{subfig:HINTCornerStruc}), instead of a single point in Figure \ref{subfig:SegCornerStruc}. The side length of a square is the smallest bucket length (i.e., the bucket length of a leaf node), which is 2 in Figure \ref{subfig:HINTCornerStruc}. For example, bucket $b_{8}$ (corresponding to range $[0,16)$) stores duplicates of intervals with $t_s \leq 0+2=2$ and $t_e \geq 16-2=14$.
An interval $s_7=[2,23)$ incurs only two copies: $[2,16)$ in $b_8$ and $[16,23)$ in $b_{20}$.
Figure \ref{subfig:HINTCornerQuery} shows the processing of a range query $[10, 13]$. HINT reports directly the intervals of green buckets (i.e., $b_{8}$, $b_{10}$, $b_{10}^R$, $b_{12}$, $b_{16}$) and requires examination of grey buckets (i.e., $b_{14}^L$, $b_{14}$).

\begin{figure}[t]
	\centering
	\begin{subfigure}{0.49\linewidth}
		\includegraphics[width=\linewidth]{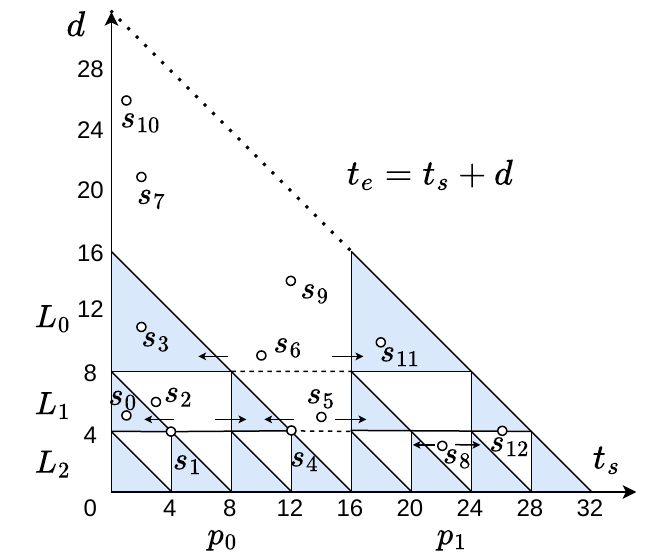}
		\caption{the structure}
		\label{subfig:PICornerStruc}
	\end{subfigure}%
	\begin{subfigure}{0.49\linewidth}
		\includegraphics[width=\linewidth]{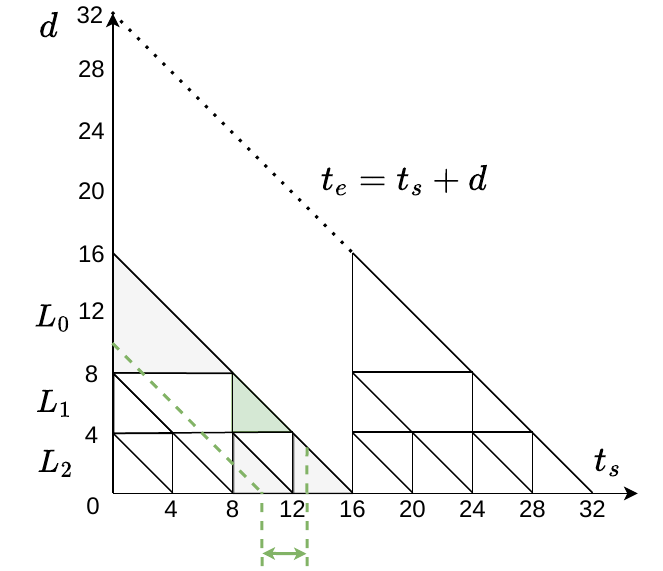}
		\caption{query with range $[10,13]$}
		\label{subfig:PICornerQuery}
	\end{subfigure}%
	\caption{Corner structure of the period-index}
	\label{fig:PeriodIndexCorner}
\end{figure}

Since the period-index assigns intervals to periods based on their duration, it is better represented by a 2D space, with $t_s$ as the $x$-axis and $d$ as the $y$-axis. Figure \ref{subfig:PICornerStruc} maps the periods of Figure \ref{fig:PeriodIndex} to triangles in a corner structure.
Period $p_0$ ($p_1$) corresponds to the triangular space with $t_s,d\in [0,16)$ ($t_s,d \in [16,32)$). Horizontal lines for duration 8 and 4 subdivide $p_0$ and $p_1$ into three levels $L_0$, $L_1$ and $L_2$. Intervals in $L_0$ must have duration $d > 8$. Accordingly, the valid space of $p_0$ ($p_1$) in $L_0$ is restricted to the upper triangle $t_s \in [0,8)$ ($[16,24)$) and $d>8$, shaded in blue. Similarly, intervals in $L_1$ must have duration in the range $(4,8]$, restricting the valid space of the corresponding buckets to the blue triangles. Intervals falling in a valid space (e.g., $s_0$, $s_3$) are stored directly in the corresponding bucket, whereas the rest (e.g., $s_2$, $s_5$, $s_6$) generate duplicates in adjacent buckets at the same level. For example, $s_2$ and $s_5$ are duplicated at $L_1$, while $s_{6}$ is duplicated at $L_0$, and stored at both $p_0$ and $p_1$.
Figure \ref{subfig:PICornerQuery} shows a query with range $[10, 13]$. Intervals in the green partition ($[8,16)$) are directly reported, while those in grey buckets ($[8,12)$) require inspection. It is worth mentioning that the original period-index \cite{Behrend2019PeriodIdx} does not differentiate between the two result types, missing an optimization opportunity.

As shown in Figure \ref{fig:RITCorner}, the RI-tree \cite{Kriegel2000RItree} decomposes the domain (0,32) into equi-length buckets at each level. Therefore, each bucket is mapped to a square-shaped area (instead of a rectangle).  The mapping only transforms buckets from the top three levels, because the lowest two levels never receive insertions and require no search.
A range query of $[10, 13]$ searches the root bucket $b_{16}$ and then $b_8$, $b_{12}$ in a top-down manner. Buckets (e.g., $b_{16}$) that intersect the boundary $t_s \leq 13$ search B+-tree $B^s$ (Figure \ref{subfig:RITCornerTs}), and buckets (e.g., $b_{8}$) that intersect the boundary $t_e \geq 10$ searches $B^e$ (Figure \ref{subfig:RITCornerTe}). Buckets that are fully covered by the query range (e.g., $b_{12}$) search either $B^s$ or $B^e$.

\begin{figure}[t]
	\centering
	\begin{subfigure}{0.49\linewidth}
		\includegraphics[width=\linewidth]{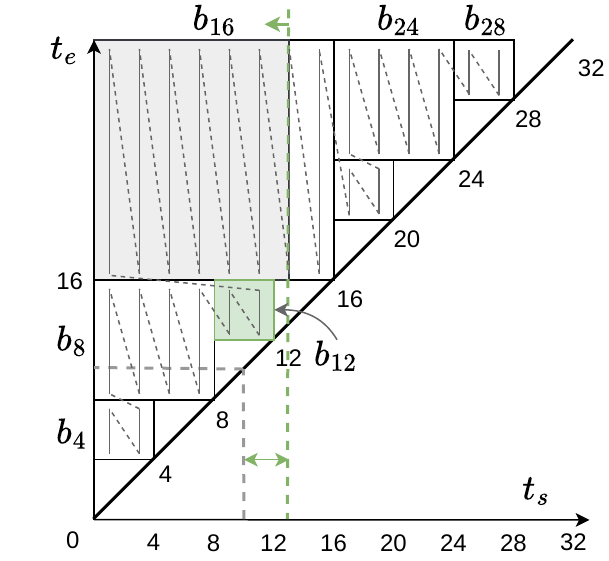}
		\caption{query $B^s$ on $(bucket, t_s)$}
		\label{subfig:RITCornerTs}
	\end{subfigure}%
	\begin{subfigure}{0.49\linewidth}
		\includegraphics[width=\linewidth]{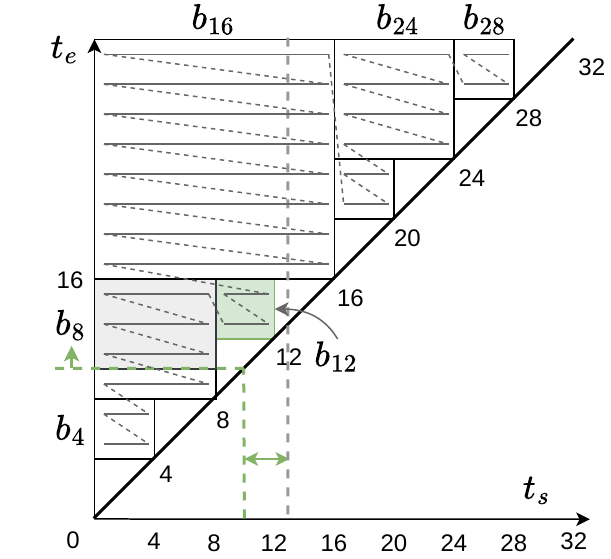}
		\caption{query $B^e$ on $(bucket, t_e)$}
		\label{subfig:RITCornerTe}
	\end{subfigure}%
	\caption{Corner structure of the RI-tree}
	\label{fig:RITCorner}
\end{figure}

With the exception of the RI-tree, the above indexes lack a maximum (or minimum) capacity constraint. Thus, their nodes or buckets may be very unbalanced.
This is particularly true for the interval-tree, where top level nodes (e.g., $n_{14}$ in Figure \ref{fig:ItvCorner}) may include most intervals. In addition, each node may contain intervals with large duration variance. The other structures alleviate these problems at the expense of redundancy.
In addition to its negative effect on index size, redundancy necessitates duplicate elimination, increasing the cost and complexity of query processing.

Interval- and segment-trees aim at static intervals, which must be known in advance. HINT supports insertions of new intervals, but the domain and the number of partition levels are predefined and fixed. Methods based on spatial indexes (e.g., \cite{Eltabakh2006SPGiST,Faloutsos1991DOT}) also assume a fixed temporal domain. One-level interval indexes (e.g., IST \cite{Goh1996IST}, the RI-tree \cite{Kriegel2000RItree}) can handle ever increasing time, but insertions may occur in arbitrary data nodes, i.e., all nodes are mutable. Numerous mutable nodes have negative effects on performance because they are: (i) under-utilized in terms of capacity (increasing the total index size), (ii) accessed during queries because they are open-ended.

If intervals arrive in increasing ending time (IET), the only way to maximize immutable nodes and achieve append-only insertions is through a two-level structure, such as SEB \cite{Song2003SEB}, with increasing $t_e$ at the bottom layer.
However, SEB suffers from a serious problem: the last node of every bottom tree (e.g., $D_5$, $D_4$ in Figure \ref{subfig:SEBstruc}) is mutable since it could receive an interval with very low $t_s$ and long duration. As shown in our experimental evaluation, for real datasets there are thousands or millions of such nodes.   In the following, we utilize our observations, to develop improved indexes for temporal intervals.

\section{CEB} \label{sec:methodCEB}
We first propose \textbf{CEB} (\textbf{C}enter and \textbf{E}ndpoint \textbf{B}-tree), which aims at improving SEB by decreasing the mutable nodes.
CEB consists of two layers of append-only B+-trees. The \emph{top tree} is ordered by the center point ($c=\frac{t_s+t_e}{2}$) of intervals.  The leaf nodes of the top tree correspond to the root nodes of \emph{bottom trees} (BTs), which are B+-trees ordered by $t_e$.
Finally, the leaves of BTs are \emph{data nodes} that store the records (interval + payload).
Each leaf node of the top and bottom trees stores a pointer to its next sibling. Under the IET assumption, a new interval is inserted into the latest leaf node of the BT that corresponds to its center point.
A BT is bounded by $(key^{min}_{top}, key^{max}_{top}]$ (range of center points). The advantage compared to SEB is that the BTs with $key^{max}_{top}< \frac{now}{2}$ will never receive insertions and are compressed to full, immutable nodes, whereas, in SEB the last node of every BT is mutable.

Figure \ref{subfig:CEBnode} shows an example CEB assuming that the data node capacity is 3. The root node of the top tree has four leaf nodes pointing to bottom tree roots $BR_0$, $BR_1$, $BR_2$, $BR_3$ separated by center point keys 6, 16.5 and 23.5. Specifically, all intervals with $c \leq 6$ fall under $BR_0$, with $c\in (6,16.5]$ fall under $BR_1$ etc. $BR_1$ separates its data nodes by endpoint $t_e=19$. A new interval (with $t_e \geq 30$) may only be inserted into $D_3$, $D_2$, or $BR_3$. $D_0$ can never receive insertions because the current time $30 > 2\cdot 6$ (twice its maximum center point).
Figure \ref{subfig:CEBstruc} shows the corner structure, where the $x$-axis is $c$, and the $y$-axis is $t_e$. All intervals appear above the diagonal line because $t_e>c$. Four bottom trees $BT_0$, $BT_1$, $BT_2$ and $BT_3$ correspond to four adjacent columns on the $x$-axis, separately by 6, 16.5 and 23.5. Their respective data nodes are rectangular partitions on the $t_e$-axis. Each interval maps to a point into a data node based on its $c$ and $t_e$. For instance, intervals with $c\in (6,16.5]$ are mapped to points in $D_1$ or $D_3$ (of $BT_1$), with $D_1$ storing older, and $D_3$ more recent intervals. The gray data space above $D_0$ will never receive insertions.

\begin{figure}[t]
	\centering
	\begin{subfigure}{0.55\linewidth}
		\includegraphics[width=\linewidth]{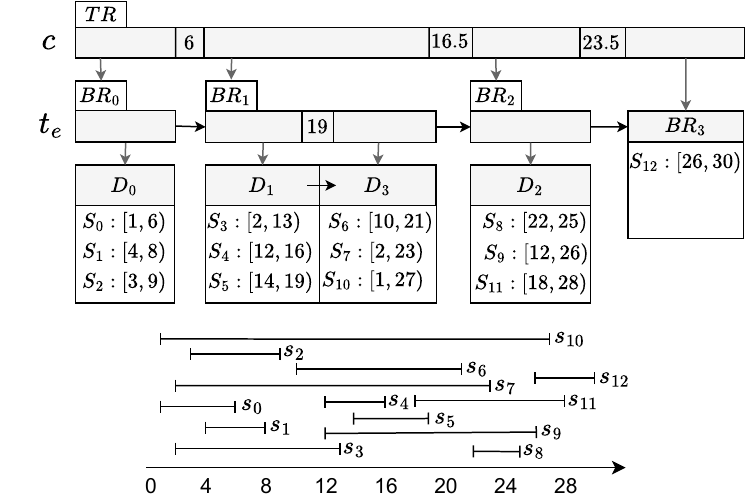}
		\caption{nodes}
		\label{subfig:CEBnode}
	\end{subfigure}%
	\begin{subfigure}{0.45\linewidth}
		\includegraphics[width=\linewidth]{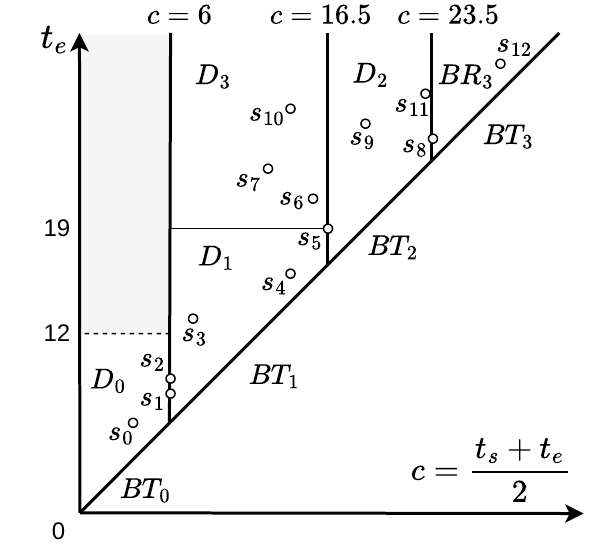}
		\caption{the structure}
		\label{subfig:CEBstruc}
	\end{subfigure}%
	\caption{CEB}
	\label{fig:CEB}
\end{figure}

\subsection{Insertions} \label{subsec:CEBInsert}
We adopt the AP-tree \cite{Segev1993APtree} to facilitate \emph{append-only} insertions for top tree and bottom trees. Each root  ($TR$ or $BR$) stores a pointer to its last leaf node, where insertions may occur. Algorithm \ref{alg:insertCEB} describes insertion of interval $I$ and its payload $P$.
First, CEB searches the top tree $TT$ for the bottom tree root $BR$ containing the center $key_{top}$ of $I$ (line 5). Then $BR$ returns the last data node $DN$ of its bottom tree $BT$. If $DN$ is not full, CEB appends $I$ and $P$ into $DN$ (line 8) and the insertion terminates. Otherwise, $DN$ requires overflow processing, generating a new data node (lines 12-24). There are three cases of splits:
\begin{itemize}
	\item Case 1: horizontal split only, in lines 13-16.
	\item Case 2: vertical split only, in lines 17-21.
	\item Case 3: vertical and horizontal split, in lines 17-24.
\end{itemize}
If the overflowed node $DN$ is not a bottom tree root $BR$, CEB splits horizontally at the largest ending time $key^{new}_{te}$, generating data node $DN^{new}$ for the insertion (Case 1). If the overflowed node $DN$ is a bottom tree root $BR$ (i.e., $BT$ has a single data node $DN$), CEB splits vertically at the maximum center point $key_{top}^{new}$, creating $BT^{new}$ with root $BR^{new}$ (lines 19-20). $BR^{new}$ is also a data node, storing intervals with center points in $c \geq key_{top}^{new}$. Given that only the latest BT root can be a data node, $BR$ is upgraded from a leaf to a branch node (lines 17-18), which has a single child. If $I.c > key_{top}^{new}$, the insertion falls into $BR^{new}$ of $BT_{new}$ (Case 2).  Otherwise (Case 3), CEB requires another horizontal split at the current largest ending time $key_{te}^{new}$, creating data node $DN_{new}$ in $BT$ (lines 22-24) for insertion.

\begin{algorithm}[!thb]
	\begin{algorithmic}[1]
		\Procedure{insert{\_}record}{$I$, $P$}
		\If{$TT$ is empty}
		\State Initialize the top tree $TT$
		\EndIf
		\State $key_{top} \gets (I.t_s + I.t_e)/2$ \Comment{Compute the top key (center)}
		\State $BR \gets$ find the bottom root corresponding to $key_{top}$ in $TT$
		\State $DN \gets$ latest data node in the bottom tree of $BR$
		\If{$DN$ is not full}
		\State Insert $(I, P)$ to $DN$
		\Else \Comment{$DN$ is full}
		\State $DN^{new} \gets \textsc{process{\_}overflow}(BR, DN, key_{top})$
		\State Insert $(I, P)$ to $DN^{new}$
		\EndIf
		\EndProcedure

		\Procedure{process{\_}overflow}{$BR$, $DN$, $key_{top}$}
		\If{$DN \ne BR$}
		\Comment{Case 1}
		\State $key^{new}_{te} \gets$ ending time of the last interval in $BR$
		\State $DN^{new} \gets$ add new data node to $BR$'s tree with $key^{new}_{te}$
		\State \Return $DN^{new}$
		\EndIf

		\Comment{If $DN$ is also the bottom tree root}
		\State $DN' \gets$ transfer the data in $DN$ to another page
		\State Upgrade $BR$ to a branch node pointing to $DN'$

		\State $key^{new}_{top} \gets$ maximum top key in $DN$
		\State $BR^{new} \gets$ add new leaf node to $TT$ with $key^{new}_{top}$

		\If{$key^{new}_{top} < key_{top}$}
		\Return $BR^{new}$ \Comment{Case 2}
		\EndIf

		\Comment{Case 3}
		\State $key^{new}_{te} \gets$ ending time of the last interval in $BR$
		\State $DN^{new} \gets$ add new data node to $BR$'s tree with $key^{new}_{te}$
		\State \Return $DN^{new}$

		\EndProcedure
	\end{algorithmic}
	\caption{Inserting a new record (I: Interval, P: Payload)}
	\label{alg:insertCEB}
\end{algorithm}

\begin{figure*}[t]
	\centering
	\begin{subfigure}{0.245\linewidth}
		\includegraphics[width=\linewidth]{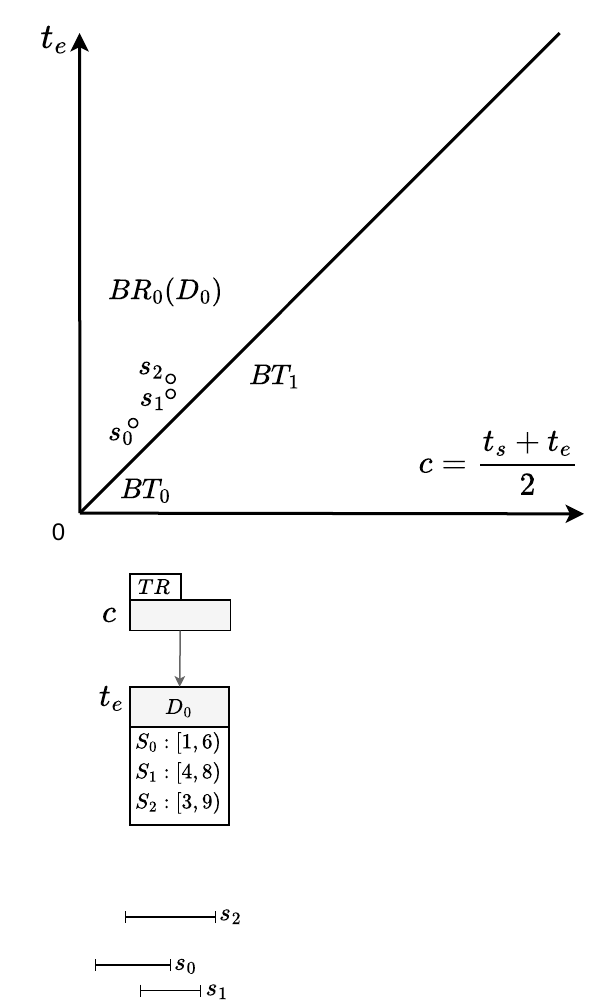}
		\caption{3-rd insertion}
		\label{subfig:CEBins3}
	\end{subfigure}%
	\begin{subfigure}{0.24\linewidth}
		\includegraphics[width=\linewidth]{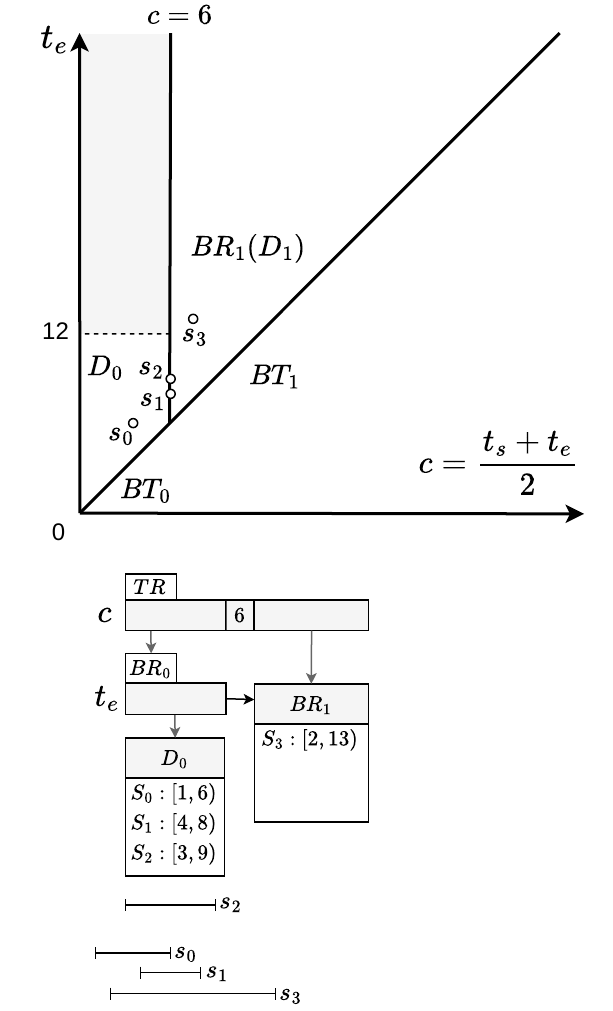}
		\caption{4-th insertion}
		\label{subfig:CEBins4}
	\end{subfigure}%
	\begin{subfigure}{0.24\linewidth}
		\includegraphics[width=\linewidth]{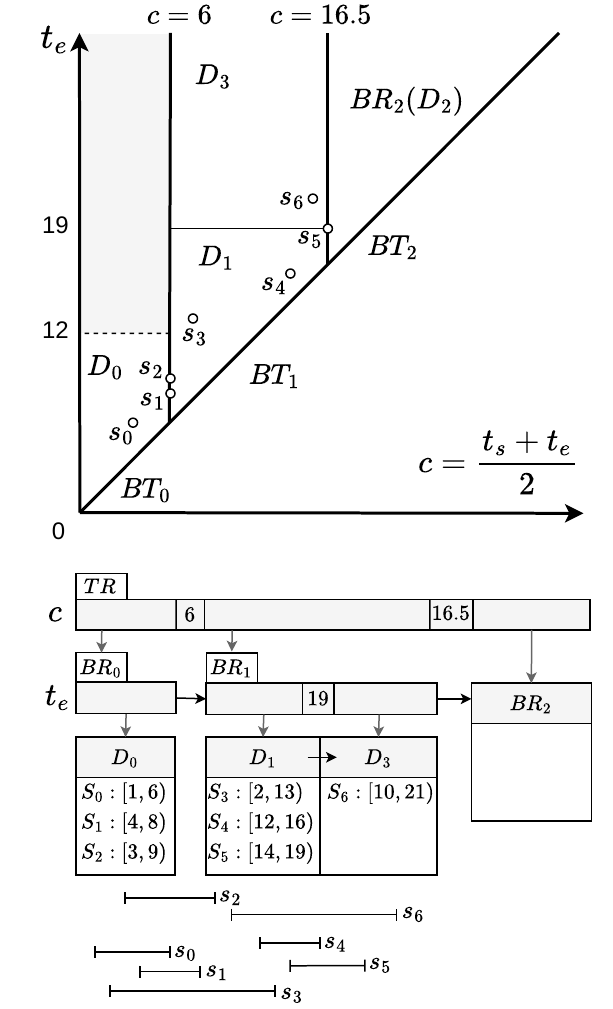}
		\caption{7-th insertion}
		\label{subfig:CEBins7}
	\end{subfigure}%
	\begin{subfigure}{0.24\linewidth}
		\includegraphics[width=\linewidth]{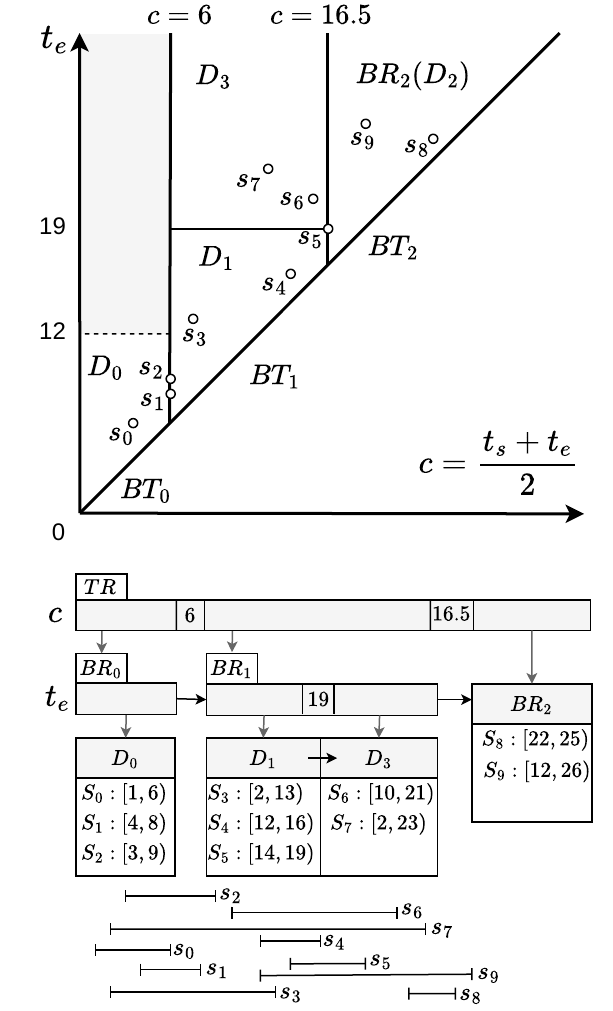}
		\caption{10-th insertion}
		\label{subfig:CEBins10}
	\end{subfigure}
	\caption{Insertions of CEB}
	\label{fig:CEBInsert}
\end{figure*}

Figure \ref{fig:CEBInsert} illustrates the insertion of the ten intervals $s_0,...,s_{9}$ of Figure \ref{fig:CEB} into CEB (assuming data node capacity 3).
Initially (Figure \ref{subfig:CEBins3}), the top tree root $TR$ has a single child $BR_0$, which is both a BT root and a data node $D_0$, containing $s_0$, $s_1$ and $s_2$. In Figure  \ref{subfig:CEBins4}, the next insertion $s_3$ forces $D_0$ to overflow at current maximum center point 6 (i.e., center of $s_2$),
generating $BT_1$ with a single data node $D_1$, which is also a root $BR_1$  (Case 2). Interval $s_3$ is inserted into $D_1$ because its center point $c>6$. $BR_0$ is upgraded from leaf to branch node, pointing to $D_0$.
$BT_0$ with center range $c\in (0,6]$ becomes immutable because $now=13 \geq 6 \cdot 2$, i.e., $BT_0$ never receives new points.
In Figure \ref{subfig:CEBins7}, $s_4$ and $s_5$ both have center point above 6 and should be inserted into the latest data node under $BR_1$.
Inserting $s_6$ causes $D_1$ to overflow at its current largest $c=16.5$, generating $BT_2$ with a single data node $D_2$.
Meanwhile, $BR_1$ is upgraded to a branch node, pointing to immutable $D_1$. Since $s_6$ has $c \leq 16.5$ and should be inserted into $BT_1$, $D_1$ overflows at its current largest $t_e=19$ (Case 3), generating a new data node $D_3$ to accommodate $s_6$.
In Figure \ref{subfig:CEBins10}, new insertions either fall into $D_3$ (e.g., $s_7$) or $D_2$ (e.g., $s_8$ and $s_9$) according to their center points.
The following Lemma \ref{lemma:ourInsert} analyzes the insertion cost of CEB.

\begin{lemma} \label{lemma:ourInsert}
	When there are $N$ intervals in $m$ bottom trees,
	a CEB insertion has I/O cost $O(\log_B m + \frac{1}{B}\cdot \log_B N + \frac{m}{N}\cdot\log_B m)$, where $B$ is the capacity of a disk page.
\end{lemma}

\begin{proof}
	Each insertion always incurs cost $C_{Search}$ to find the proper leaf node.
	Locating the bottom tree, is bounded by the depth of the top tree $O(\log_B m)$.
	Since each bottom tree is append-only, fetching its last node costs only $O(1)$.
	Therefore, the cost of insertions in the absence of overflows is dominated by the search cost $C_{Search}=O(\log_B m)$.
	In addition, overflows require a horizontal split (Case 1) with cost $C_{HS}$, or vertical split (Case 2) with cost $C_{VS}$ or both (Case 3) with cost $C_{HS}+C_{VS}$.
	A horizontal split inserts a new data node into a bottom tree\footnote{A horizontal split has no cost for the top tree.}, and may propagate all the way up to its root with cost bounded by the depth of the largest bottom tree $O(\log_B N)$.
	In the worst case (Case 1\&3), a horizontal split occurs after every $O(B)$ insertions; thus,
	the amortized $C_{HS}= O(\frac{1}{B}\cdot \log_B N)$.
	Moreover, an insertion may incur a vertical split with a probability $\frac{m}{N}$, generating a single-node bottom tree, which requires an insertion into the top tree with $O(\log_B m)$ cost.
	Therefore, the amortized cost $C_{VS}= O(\frac{m}{N} \cdot \log_B m)$. Adding the three terms together: $C_{Search}+C_{HS}+C_{VS}=O(\log_B m + \frac{1}{B}\cdot \log_B N + \frac{m}{N}\cdot\log_B m)$.
\end{proof}

As an additional step to reduce the number of bottom trees, CEB conducts a periodic check that identifies branches of the top tree having only immutable BTs, and compacts them into a single balanced B+-Tree containing all points (of the immutable BTs). Since the data nodes of individual BTs are already ordered on $t_e$, we merge them to form a single sorted linked list of dead data nodes. We then iterate through this list packing every $C_B$ data nodes in a new branch, where $C_B$ is the branch capacity. We repeat recursively to create higher level branches in a bottom-up routine. 
This local bulk loading process replaces $C_B$ immutable BTs (whose last data node can be almost empty), with a single B+-tree where all nodes, except possibly for the last one, are full. The compaction step is not applicable to SEB because all its BTs are mutable.

\subsection{Range Queries} \label{sec:CEBRangeQ}
Given a range query $q=[q_s,q_e]$, CEB returns intervals fulfilling constraints $t_s=2c-t_e \leq q_e$ (slash boundary) and $t_e\geq q_s$ (horizontal boundary).
Notably, the query searches bottom trees with center points in $[ \frac{q_s}{2}, \frac{q_e+now}{2}]$, which corresponds to the green area in Figure \ref{fig:CEBQuery}.
Observe that center points within the triangular region bounded by $q_e$ and $\frac{q_e+now}{2}$ may correspond to qualifying intervals with $t_s \leq q_e$.
On the other hand, bottom trees with (i) maximum center point $key^{max}_{top}$ before $\frac{q_s}{2}$ contain (unqualified) intervals that end before $q_s$, and (ii) minimum center point $key^{min}_{top}$ after $\frac{q_e+now}{2}$ have intervals that start after $q_e$.
In Figure \ref{fig:CEBQuery}, $BT_0$ and $BT_1$ are excluded by case (i), while $BT_6$ by case (ii).
In the rest of the bottom trees ($BT_2$ to $BT_5$)  intervals in nodes included by the range (e.g., $D_7$, $D_8$, $D_{10}$, $D_{11}$) are directly reported, while those in nodes partially overlapped (e.g., $D_4$, $D_5$, $D_6$, $D_9$) are evaluated.

\begin{figure}[t]
	\centering
	\includegraphics[width=0.75\linewidth]{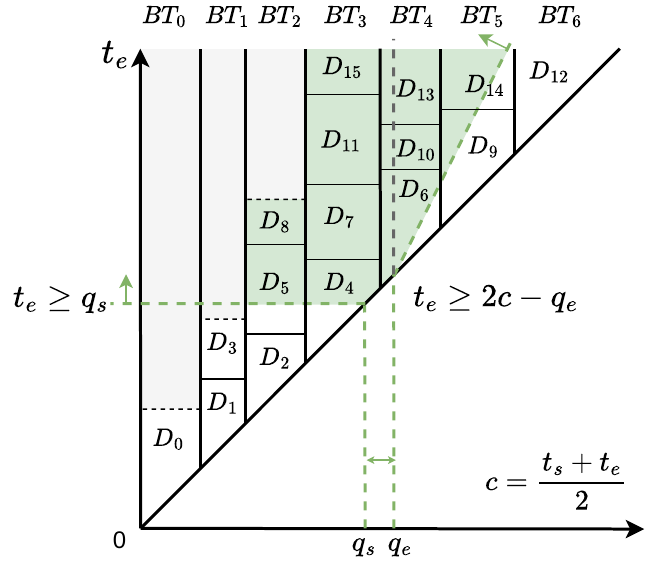}
	\caption{Range query $[q_s, q_e]$ in CEB}
	\label{fig:CEBQuery}
\end{figure}

Algorithm \ref{alg:queryCEB} shows the pseudocode for range query processing $[q_s,q_e]$.
CEB accesses bottom trees with center points ranging from $\frac{q_s}{2}$ (line 3) to $\frac{q_e+now}{2}$ (line 5).
Searching a BT starts from the data node containing $q_s$ (horizontal boundary) or $2\cdot key^{min}_{top}-q_e$ (slash boundary) (line 7), until reaching the last data node (line 8).
Intervals of the first node must be individually examined (lines 11-12), while those of the remaining nodes are directly reported (lines 13-14).

\begin{algorithm}[t]
	\begin{algorithmic}[1]
		\Procedure{range{\_}query}{$[q_s, q_e]$}
		\State $R \gets \{\}$ \Comment{Result Intervals}
		\State $BT \gets$ find the bottom tree containing center point $\frac{q_s}{2}$
		\State $key^{min}_{top} \gets$ minimum top key in $BT$
		\While{ $key^{min}_{top} \leq \frac{q_e+now}{2} $}
		\State $key^{max}_{top} \gets$ maximum top key in $BT$
		\State $DN \gets$ find the data node for $max(q_s,2\cdot key^{min}_{top}-q_e)$ in $BT$
		\While{$DN \neq$ null}
		\State $key^{min}_{te} \gets$ first key of $DN$
		\State $key^{max}_{te} \gets$ last key of $DN$
		\If{$key^{min}_{te} < max(q_s,2\cdot key^{max}_{top}-q_e)$}
		\State Append qualifying records of $DN$ to $R$
		\Else
		\State Append all records of $DN$ to $R$ \Comment{Report}
		\EndIf
		\State $DN \gets$ next $DN$
		\EndWhile
		\State $key^{min}_{top} \gets key^{max}_{top}$
		\State $BT \gets$ next $BT$
		\EndWhile
		\State \Return $R$
		\EndProcedure
	\end{algorithmic}
	\caption{Searching a range ($[q_s, q_e]$) in CEB}\label{alg:queryCEB}
\end{algorithm}

\begin{lemma} \label{lemma:CEBQuery}
	When there are $N$ intervals in $m$ bottom trees,
	a stabbing/range query of CEB has I/O cost $O(\log_B m+m\log_B N+\frac{k}{B})$,
	where $k$ is the number of query results.
\end{lemma}

\begin{proof}
	Given a range query $[q_s,q_e]$, CEB locates the first bottom tree with cost $O(\log_B m)$.
	In each BT, finding the first intersected data node has cost $O(\log_B N)$.
	The total number of data nodes containing $k$ results is $\frac{k}{B}$, and all these nodes are retrieved.   Combining the terms, we obtain the total cost as  $O(\log_B m+m\log_B N+\frac{k}{B})$.
\end{proof}

CEB can easily be extended to \emph{count} queries, returning only the number of intersected intervals, as opposed to their IDs.
In this case, intervals in immutable nodes fully covered by the range (i.e., $D_7$, $D_{10}$, $D_{11}$ of Figure \ref{fig:CEBQuery}) are directly aggregated to the total count, without requiring disk accesses. Since such nodes are fully packed, their number of intervals is fixed. Disk accesses are only necessary for partially overlapped (e.g., $D_4$, $D_5$, $D_6$) or mutable nodes (e.g., $D_8$, $D_{13}$).

\section{TIDE} \label{sec:methodTIDE}
\textbf{TIDE} (\textbf{t}ime \textbf{i}ntervals by \textbf{d}uration and \textbf{e}ndpoint) aims at minimizing the number of bottom trees, by ordering the top tree on duration $d$, and BTs on ending time $t_e$.
Duration is a property that does not necessarily increase as time evolves. This is particularly true for \emph{regular} datasets, such as transportation (e.g., taxi trips, flights) intervals, where $d$ is independent of the length of the recorded history. For instance, the average flight time (e.g., a few hours) does not increase based on the years of stored flights. As another example, intervals produced by IOT devices or wireless sensors have a fixed length determined by the sampling or transmission frequency. 
Especially for regular datasets, TIDE leads to a very small number of bottom trees (several orders of magnitude below SEB and CEB), facilitating high compactness and cache locality (i.e., most non-full nodes are cached).

\begin{figure}[!thb] 
	\centering
	\begin{subfigure}{0.55\linewidth}
		\includegraphics[width=\linewidth]{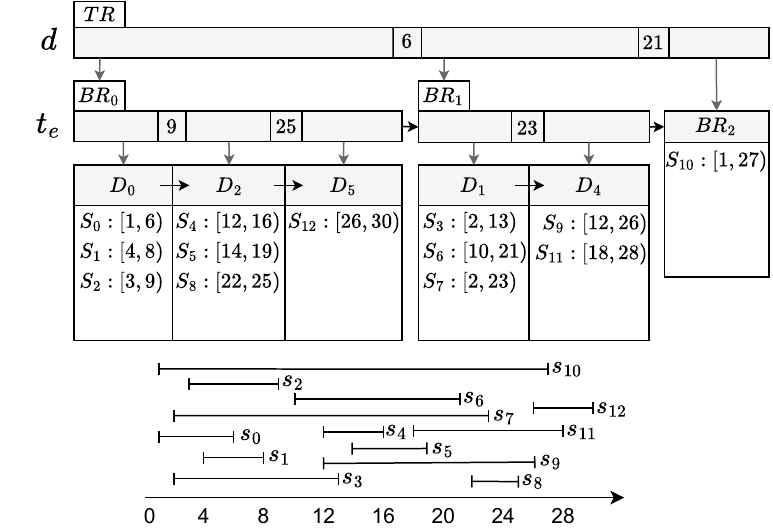}
		\caption{nodes}
		\label{subfig:HT}
	\end{subfigure}%
	\begin{subfigure}{0.45\linewidth}
		\includegraphics[width=\linewidth]{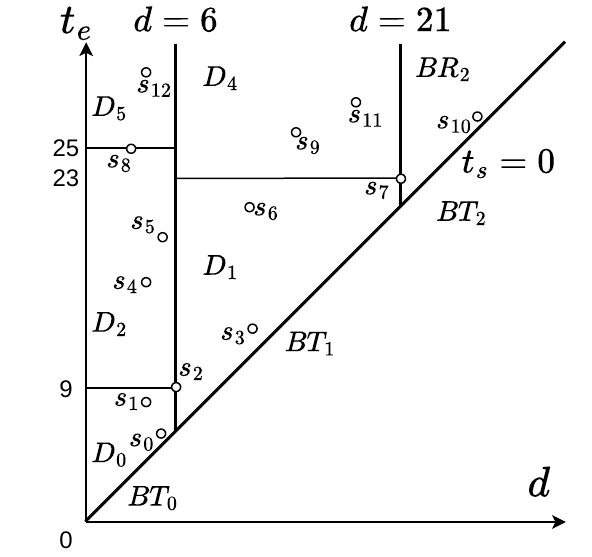}
		\caption{corner structure}
		\label{subfig:HTUniversal}
	\end{subfigure}%
	\caption{TIDE}
	\label{fig:TIDE}
\end{figure}

Figure \ref{subfig:HT} shows an example of TIDE storing thirteen intervals (same examples in Figure \ref{fig:CEB}).
The root node of the top tree $TR$ contains three leaf nodes, corresponding to three bottom tree roots $BR_0$, $BR_1$ and $BR_2$, separated by duration keys 6 and 21. $BR_0$ ($BR_1$) separates its data nodes with ending time keys 9 and 25 (23). $BR_2$ is both a root and a data node. A new interval (with $t_e \geq 30$) may only be inserted at $D_5$, $D_4$ or $BR_2$. Figure \ref{subfig:HTUniversal} shows the corresponding corner structure in the duration and end-time 2D space. All intervals appear on the upper half of the diagonal $t_e=d$. The three bottom trees correspond to three adjacent columns on the $d$-axis, separated by 6 and 21.
Each interval maps to a point into a data node based on its $d$ and $t_e$. For instance, intervals with $d \in (6,21]$ are mapped to points in $D_1$ or $D_4$ (of $BT_1$), with $D_1$ storing older, and $D_4$ more recent intervals.

\begin{figure*}[b]
	\centering
	\begin{subfigure}{0.235\linewidth}
		\includegraphics[width=\linewidth]{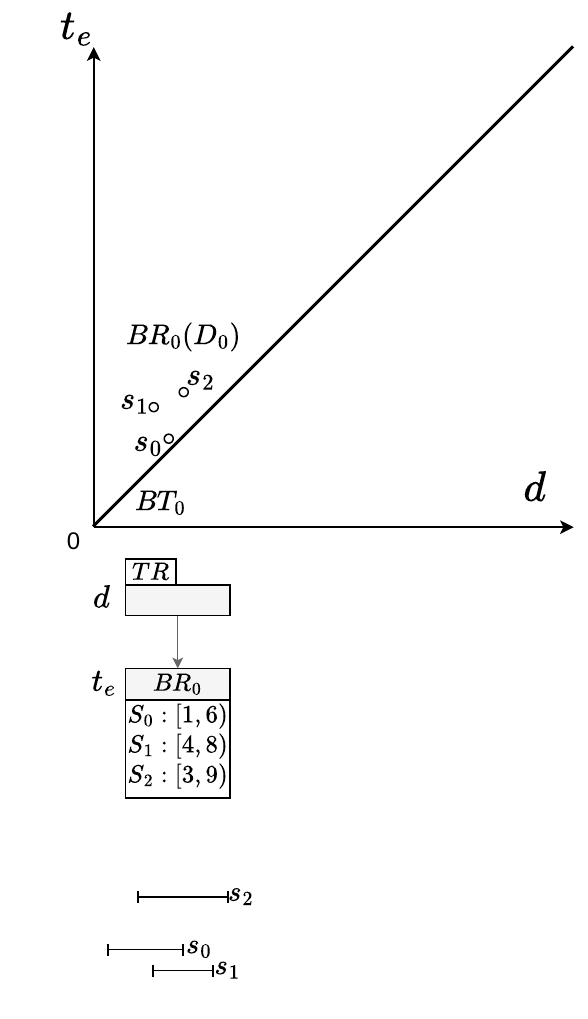}
		\caption{3-rd insertion}
		\label{subfig:TIDEins3}
	\end{subfigure}%
	\begin{subfigure}{0.235\linewidth}
		\includegraphics[width=\linewidth]{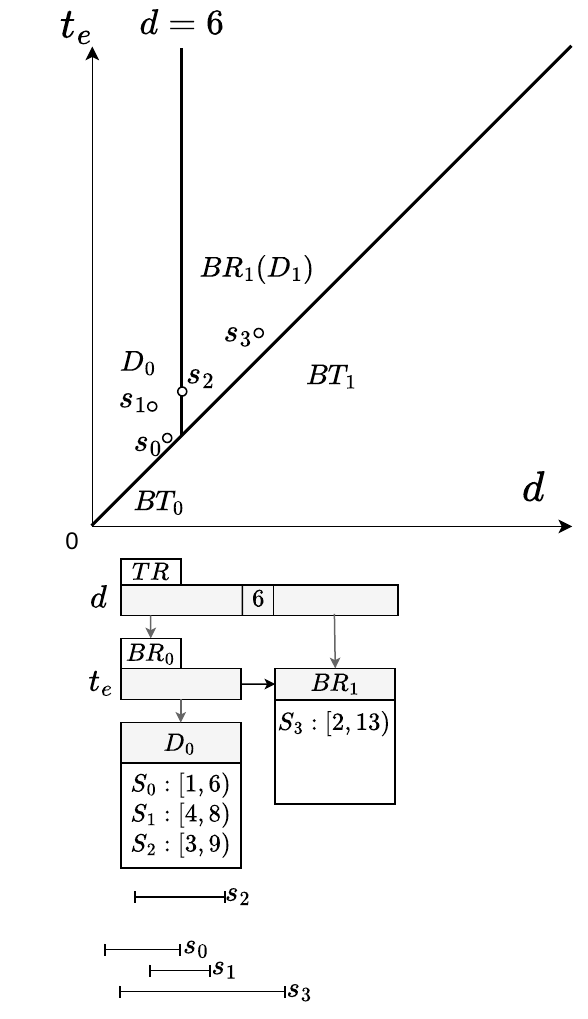}
		\caption{4-th insertion}
		\label{subfig:TIDEins4}
	\end{subfigure}%
	\begin{subfigure}{0.235\linewidth}
		\includegraphics[width=\linewidth]{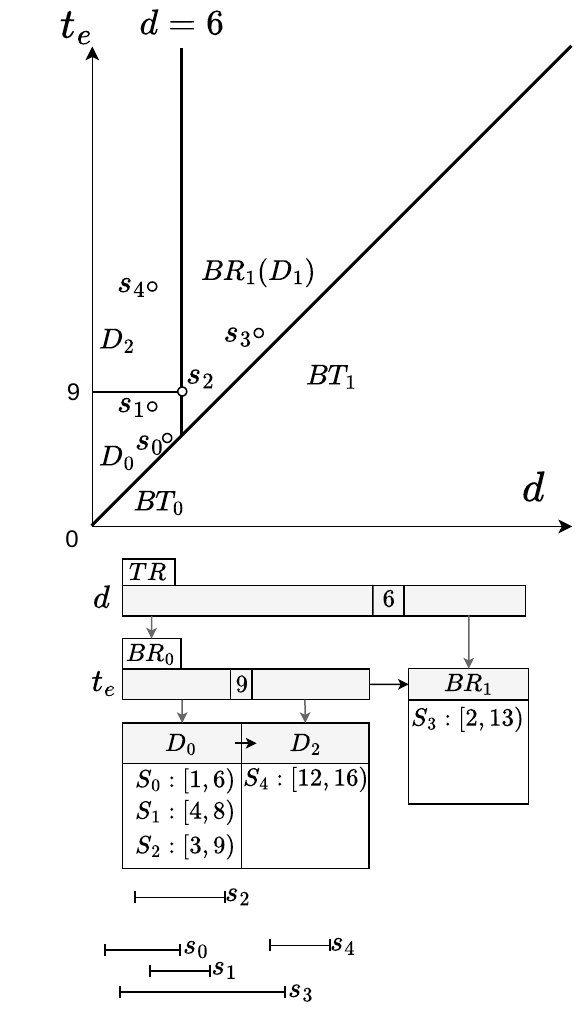}
		\caption{5-th insertion}
		\label{subfig:TIDEins5}
	\end{subfigure}%
	\begin{subfigure}{0.29\linewidth}
		\includegraphics[width=\linewidth]{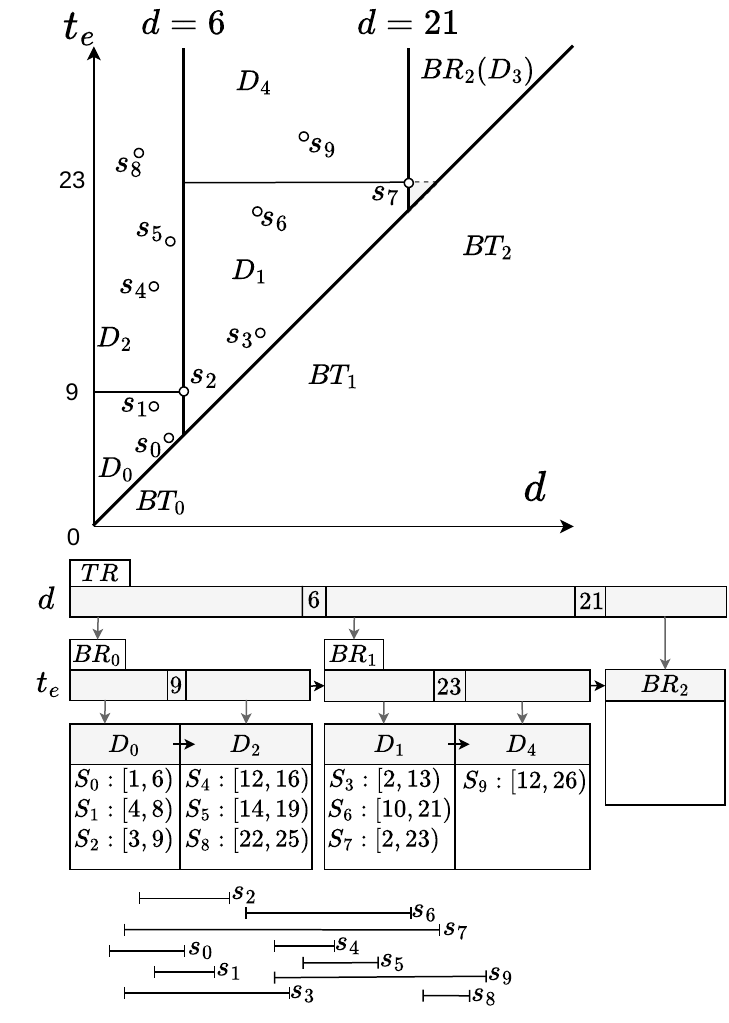}
		\caption{10-th insertion}
		\label{subfig:TIDEins10}
	\end{subfigure}
	\caption{Insertions in TIDE}
	\label{fig:TIDEInsert}
\end{figure*}

\subsection{Insertions}  \label{subsec:TIDEInsert}
Insertions of TIDE are similar to CEB, except for computing the top key. Specifically, line 4 in Algorithm \ref{alg:insertCEB} is replaced with $key_{top} \gets I.t_e - I.t_s$. The rest of the process is unchanged.
Figure \ref{fig:TIDEInsert} illustrates the insertion of the ten intervals $s_0,...,s_{9}$ of Figure \ref{fig:TIDE} into TIDE (assuming data node capacity 3).
Initially (Figure \ref{subfig:TIDEins3}), the top tree root $TR$ has a single child $BR_0$, which is both a bottom tree root and a data node $D_0$, containing $s_0$, $s_1$ and $s_2$. In Figure  \ref{subfig:TIDEins4}, the next insertion $s_3$ forces $D_0$ to overflow at current maximum duration 6 (i.e., duration of $s_2$), generating $BT_1$ with a single data node $D_1$, which is also a BT root $BR_1$  (Case 2 in Section \ref{subsec:CEBInsert}). Interval $s_3$ is inserted into $D_1$, since it has duration $d>6$. $BR_0$ is upgraded from leaf to branch node, pointing to $D_0$. In Figure  \ref{subfig:TIDEins5}, $s_4$ has duration below 6 and should be inserted into the latest data node under $BR_0$, forcing $D_0$ to overflow at its maximum $t_e=9$ (Case 1). A new data node $D_2$ accommodates $s_4$, and $D_0$ becomes immutable. Future insertions either fall into $D_1$ (e.g., $s_6$ and $s_7$) or $D_2$ (e.g., $s_5$ and $s_8$) according to their duration. In Figure  \ref{subfig:TIDEins10}, inserting $s_9$ causes $D_1$ to overflow at its current largest $d=21$, generating $BT_2$ with a single data node $D_3$.
Meanwhile, $BR_1$ is upgraded to a branch node, pointing to immutable $D_1$. Since $s_9$ has $d \leq 21$ and should be inserted into $BT_1$, $D_1$ overflows at its current largest $t_e=23$ (Case 3), generating data node $D_4$ to insert $s_9$.

Theoretically, TIDE has the same insertion cost as CEB ($C_{Search}+C_{HS}+C_{VS}$ in Lemma \ref{lemma:ourInsert}).
In practice, $C_{Search}$ and $C_{VS}$ of TIDE are negligible since $m << N$, and $C_{HS}$ dominates the insertion cost of TIDE. The compaction step of CEB for immutable bottom trees is not applicable to TIDE  because all BTs may receive insertions, given that future intervals may have any duration. 

\subsection{Range Queries} \label{subsec:TIDEQuery}
Since the top tree organizes intervals by duration, it is not useful for range queries. Instead, given a range query $[q_s,q_e]$, where $q_s \leq q_e$, TIDE searches \emph{all} bottom trees and returns intervals fulfilling $t_e \geq q_s$ (horizontal boundary) and $t_s \leq q_e$ (diagonal boundary).
Figure \ref{fig:TIDEQuery} shows an example range, where the result area is shaded in green. For instance, nodes in $BT_0$ with possible results (i.e., $D_2$, $D_5$ and $D_6$) intersect $[q_s, q_e+d_1]$. Nodes below the horizontal boundary $t_e < q_s$ contain intervals that end before $q_s$, whereas nodes above the diagonal boundary $t_s > q_e$ have intervals that start after $q_e$. Similarly, in $BT_1$ nodes possibly containing results are $D_1$, $D_7$ and $D_9$. Intervals in nodes, such as $D_4$, covered by the range are directly reported.

\begin{figure}[t]
	\centering
	\includegraphics[width=0.75\linewidth]{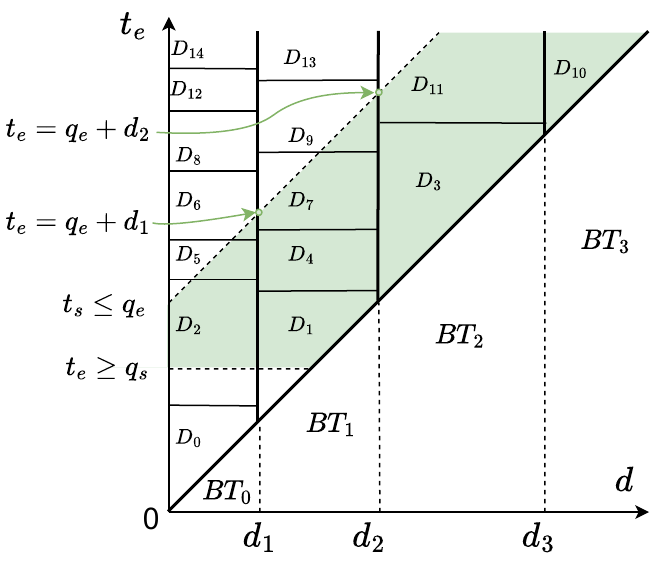}
	\caption{Range query $[q_s, q_e]$ in TIDE}
	\label{fig:TIDEQuery}
\end{figure}

Algorithm \ref{alg:queryTIDE} shows the pseudocode for range query processing $[q_s,q_e]$.
Each bottom tree $BT$ stores intervals with durations ranging from $key^{min}_{top}$ to $key^{max}_{top}$, shaped as a column.
Searching a bottom tree starts from the data node containing $q_s$ (line 7),
and stops when reaching the data node containing $q_e + key^{max}_{top}$ (lines 11-12) or the last data node (line 8).
Under the unified representation, TIDE identifies nodes that can be directly reported (i.e., all its intervals are results),
if $q_s \leq key^{min}_{te}$ and $key^{max}_{te} \leq q_e + key^{min}_{top}$ (lines 13-14).
The intervals of the remaining searched nodes may constitute results, and must be individually examined (lines 15-16).

\begin{algorithm}[!thb]
	\begin{algorithmic}[1]
		\Procedure{range{\_}query}{$[q_s, q_e]$}
		\State $R \gets \{\}$ \Comment{Result Intervals}
		\State $BT \gets$ earliest bottom tree
		\State $key^{min}_{top} \gets 0$
		\While{$BT \neq$ null}
		\State $key^{max}_{top} \gets$ maximum top key in $BT$
		\State $DN \gets$ find the data node for $q_s$ in $BT$
		\While{$DN \neq$ null}
		\State $key^{min}_{te} \gets$ first key of $DN$
		\State $key^{max}_{te} \gets$ last key of $DN$
		\If{$q_e + key^{max}_{top} < key^{min}_{te}$}
		\State \textbf{break}
		\EndIf
		\If{$q_s \leq key^{min}_{te}$ and $key^{max}_{te} \leq q_e + key^{min}_{top}$}
		\State Append all records of $DN$ to $R$ \Comment{Report}
		\Else
		\State Append qualifying records of $DN$ to $R$
		\EndIf
		\State $DN \gets$ next $DN$
		\EndWhile
		\State $key^{min}_{top} \gets key^{max}_{top}$
		\State $BT \gets$ next $BT$
		\EndWhile
		\State \Return $R$
		\EndProcedure
	\end{algorithmic}
	\caption{Searching a range ($[q_s, q_e]$) in TIDE}\label{alg:queryTIDE}
\end{algorithm}

\begin{lemma} \label{lemma:TIDEQuery}
	When there are $N$ intervals in $m$ bottom trees,
	a stabbing/range query of TIDE has I/O cost $O(m\log_B N+\frac{k}{B})$,
	where $k$ is the number of query results.
\end{lemma}

\begin{proof}
	Given a range query $[q_s,q_e]$, TIDE searches (in each bottom tree $BT_i$) for the data node containing $q_s$ with cost $O(\log_B N)$.
	Then it scans its siblings until finding the data node containing $q_e + key^{max}_d$ (i.e., the maximum duration in a bottom tree) or the last data node. The total number of data nodes containing $k$ results is $\frac{k}{B}$. Combining the search and scan terms, we obtain  $O(m\log_B N+\frac{k}{B})$.
\end{proof}

Similar to CEB, \emph{count} queries (returning only the number of intersected intervals) require no disk accesses for immutable nodes fully covered by the range (e.g., $D_3$, $D_4$ of Figure \ref{fig:TIDEQuery}). TIDE can also efficiently process queries with length constraints, e.g., find all intervals in $[q_s,q_e]$ with duration in the range $[d_s,d_e]$.
\footnote{The RI-tree cannot answer such queries because $t_s$ and $t_e$ of each interval are stored separately (i.e., interval durations are lost).}
In this case, the top tree is used to identify bottom trees with intervals satisfying $[d_s,d_e]$. Even for conventional ranges (without duration constraints), where all $m$ bottom trees are accessed, TIDE's query performance is outstanding because $m$ is very low in real-world datasets. We experimentally evaluate our claims in the next section.

\begin{table*}[!thb]
	\scriptsize
	\small
	\centering
	\begin{tabular}{c|c c c c}
		\toprule
		              & TAXI                                            & BIKE                                           & NFT                                            & AMAZON                                          \\
		\cmidrule(lr){1-5}
		Object        & NYC taxi                                        & NYC bicycle                                    & non-fungible token                             & online item                                     \\
		Interval      & driving trips                                   & riding trips                                   & unchanged price                                & unchanged rating                                \\
		\#Intervals   & 3198336346 (3.2 billion)                        & 100507903 (100.5 million)                      & 28581957  (28.6 million)                       & 439293331  (439.3 million)                      \\
		\cmidrule(lr){1-5}
		Min. Duration & 1 second                                        & 60 seconds                                     & 1 second                                       & 1 second                                        \\
		\cmidrule(lr){2-5}
		Avg. Duration & \makecell{1115 secs (0.0002\%)\\(18.6 minutes)} & \makecell{984 secs (0.0004\%)\\(16.4 minutes)} & \makecell{2724947 secs (3.2\%)\\(1 month)}     & \makecell{2789003 secs (1.0\%)\\(1 month)}      \\
		\cmidrule(lr){2-5}
		Max. Duration & \makecell{76709271 secs (16.2\%)\\(2.4 years)}  & \makecell{19513649 secs (8.8\%)\\(7.4 months)} & \makecell{83743716 secs (99.6\%)\\(2.7 years)} & \makecell{288272408 secs (99.2\%)\\(9.1 years)} \\
		\bottomrule
	\end{tabular}
	\vspace{0.3cm}
	\caption{Properties of interval datasets}
	\label{table:datasets}
\end{table*}

\section{Experimental Evaluation} \label{sec:exp}

The evaluation was conducted on Ubuntu Linux with an AMD Ryzen Threadripper 3960X 3.8GH CPU and 64GiB RAM.
We developed a generic disk-based framework for historical indexes in Rust, implemented TIDE, CEB and SEB using the same append-only B+-tree structures \cite{Segev1993APtree}, and compared them with the state-of-the-art RI-tree:
\begin{itemize}
	\item TIDE: two-level B+-trees ordered by $d$ and $t_e$
	\item CEB: two-level B+-trees ordered by $c=\frac{t_s+t_e}{2}$ and $t_e$
	\item SEB \cite{Song2003SEB}: two-level B+-trees ordered by $t_s$ and $t_e$
	\item RIT \cite{Kriegel2000RItree}: a B+-tree $B^s$ ordered by $(bucket, t_s)$ and another $B^e$ ordered by $(bucket, t_e)$.
\end{itemize}
The disk-page and LRU cache sizes are set to 4KiB and 4MiB, respectively, for all experiments (different cache sizes exhibit similar performance). We use the following real datasets, summarized in Table \ref{table:datasets}:
\begin{itemize}
	\item \textbf{TAXI \footnote{\url{https://www.nyc.gov/site/tlc/index.page}}}: 3 billion taxi trips with pick up and drop off timestamps (in seconds), during 2011-2025.
	\item \textbf{BIKE \footnote{\url{https://citibikenyc.com/system-data}}}:
	      Start and end timestamps (in seconds) of 100M bicycle trips, during 2014-2020 in New York City.
	\item \textbf{NFT \footnote{\url{https://huggingface.co/datasets/MLNTeam-Unical/NFT-70M\_transactions}}} \cite{Costa2023nft}: 28M intervals denoting the stable price period  (in seconds) of non-fungible tokens from OpenSea transactions, during 2021-2023.
	\item \textbf{AMAZON \footnote{\url{https://amazon-reviews-2023.github.io/}}} \cite{Hou2024amazon}: 439M intervals corresponding to the rating period  (in seconds) of Amazon items, during 2014-2023.
\end{itemize}

\begin{figure*}[!thb]
	\centering
	\includegraphics[width=\linewidth]{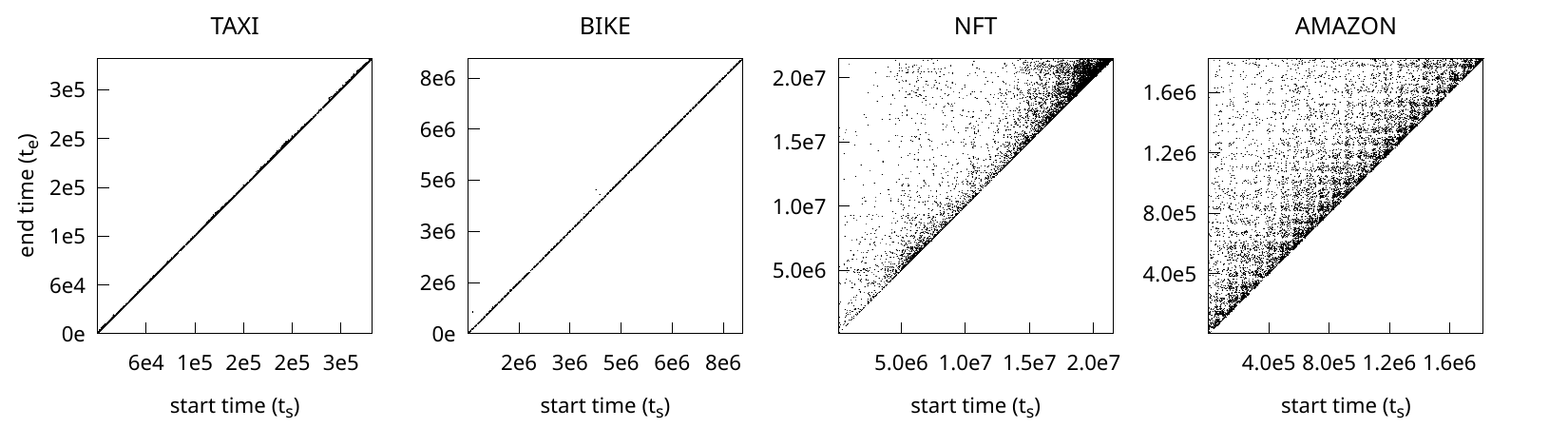}
	\caption{10k data points sampled from datasets}
	\label{fig:sample-data}
\end{figure*}

Figure \ref{fig:sample-data} plots a sample of 10k data points from the four datasets. TAXI and BIKE are \emph{regular}, i.e., interval durations are short and independent of the length of recorded history. Accordingly all data points lie close to the $t_s=t_e$ diagonal line. On the other hand, NFT and AMAZON have high variance. As shown in Table \ref{table:datasets}, the average and maximum duration for TAXI (BIKE) is only 0.0002\% (0.0004\%) and 16.2\% (8.8\%) of the whole extent of the dataset, in comparison to NFT's 3.2\% and 99.6\% (AMAZON's 1.0\% and 99.2\%). Section \ref{subsec:IndexC} investigates the size and other index characteristics, while Sections \ref{subsec:expInsert} and \ref{subsec:expQuery} evaluate insertion and query performance, respectively.

\begin{table*}[!htb]
	\centering
	\scriptsize
	\setlength{\tabcolsep}{1.8pt}
	\begin{tabular}{c|cccc|cccc|cccc|cccc}
		\toprule
		Dataset        & \multicolumn{4}{c|}{TAXI} & \multicolumn{4}{c|}{BIKE} & \multicolumn{4}{c|}{NFT} & \multicolumn{4}{c}{AMAZON}                                                                                                                                             \\
		\cmidrule(lr){2-5} \cmidrule(lr){6-9} \cmidrule(lr){10-13} \cmidrule(lr){14-17}
		               & TIDE                      & SEB                       & CEB                     & RIT                        & TIDE            & SEB    & CEB    & RIT     & TIDE            & SEB    & CEB    & RIT    & TIDE             & SEB     & CEB     & RIT     \\
		\midrule
		Size (GiB)     & \textbf{44.95}            & 51.00                     & 49.65                   & 153.92                     & \textbf{1.22}   & 2.06   & 1.87   & 4.34    & \textbf{0.91}   & 1.03   & 1.02   & 3.12   & \textbf{7.40}    & 7.80    & 7.68    & 25.49   \\
		\cmidrule(lr){1-17}
		\#Bottom Trees & \textbf{7}                & 1072630                   & 830253                  & -                          & \textbf{4}      & 145523 & 112245 & -       & \textbf{34}     & 20449  & 18870  & -      & \textbf{243}     & 72852   & 50327   & -       \\
		\cmidrule(lr){1-17}
		\#Data Nodes   & \textbf{11758592}         & 12293769                  & 12171730                & 40093046                   & \textbf{320091} & 395187 & 378364 & 1129963 & \textbf{238201} & 248306 & 247560 & 813520 & \textbf{1935334} & 1971373 & 1960372 & 6643454 \\
		\cmidrule(lr){1-17}
		\#All Nodes    & \textbf{11784554}         & 13368766                  & 13015317                & 40347903                   & \textbf{320801} & 541031 & 491184 & 1137159 & \textbf{238748} & 268800 & 266626 & 818098 & \textbf{1939736} & 2044385 & 2012627 & 6680749 \\
		\bottomrule
	\end{tabular}
	\vspace{0.2cm}
	\caption{Index size and number of nodes}
	\label{table:size}
\end{table*}

\begin{figure*}[!thb]
	\centering
	\includegraphics[width=\linewidth]{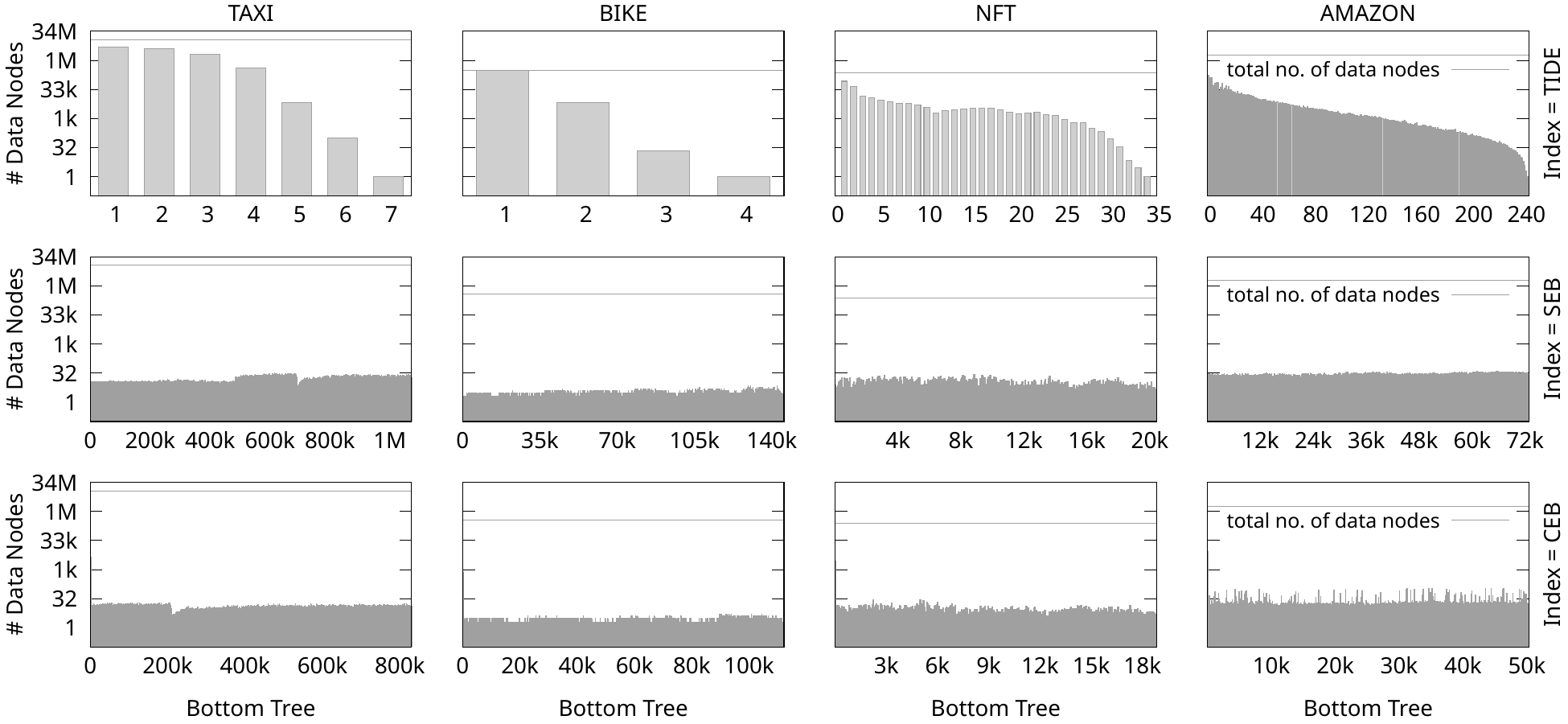}
	\caption{Number of data nodes per bottom tree}
	\label{fig:data-nodes}
\end{figure*}

\subsection{Index Characteristics} \label{subsec:IndexC}
We built indexes by sequentially inserting the records of TAXI, BIKE, NFT and AMAZON datasets.
Table \ref{table:size} lists various properties of the structures generated by TIDE, SEB, CEB and the RI-tree. 
Indexes with two-level append only B+-trees (i.e., TIDE, SEB, CEB) are significantly more compact than the RI-tree, which consumes more than three times the size of TIDE. This is because the RI-tree (i) inserts each interval into two B+-trees $B^s$ and $B^e$, generating two copies, and (ii) all its data nodes are mutable, and possibly underutilized. 
For the regular datasets TAXI and BIKE, TIDE creates only seven and four bottom trees, respectively. For the irregular, NFT and AMAZON, the number of BTs increases to 34 and 243, which however is significantly below that of SEB and CEB, which may reach up to millions. The compaction step of CEB reduces the number of BTs $8\% - 30\%$, depending on the dataset, compared to SEB.  
The index size is dominated by the data nodes in most cases. The only exception is SEB and CEB for BIKE, where SEB (CEB) data nodes only amount to 73\% (77\%) of the total, and the bottom tree roots make up nearly a quarter of all the nodes. Furthermore, each root of the mutable BTs contains only a handful of data node entries, the last of which is only partially filled. Consequently, for BIKE, SEB (CEB) consumes 1.7x (1.5x) more space than TIDE, which is fully packed, except for the last data node of its four bottom trees.

Figure \ref{fig:data-nodes} contains twelve plots, each measuring the number of data nodes ($y$-axis) stored in the corresponding bottom tree ($x$-axis) for TIDE (first row), SEB (second row) and CEB (third row). 
Each column corresponds to a dataset. The horizontal line in each diagram denotes the total number of data nodes. 
For TAXI (BIKE), the first three (one) BTs of TIDE contain the shortest 96\% (98\%) intervals, and the remaining trees contain fewer data nodes with increasing interval durations.
On the other hand, SEB and CEB have thousands or millions BTs, each with up to 30 data nodes. 
Data nodes of SEB and CEB have a similar extent in $t_e$, which is significantly higher than the average node extent in TIDE. Long nodes negatively affect performance as they are expected to intersect more queries.
For the irregular datasets, the first two bottom trees of TIDE contain the majority (56\%) of data nodes in NFT, whereas the first two BTs only occupy $16\%$ data nodes in AMAZON. In case of SEB and CEB, since the latest intervals in NFT and AMAZON are less likely to have started as recently as in the regular datasets, they are more prone to fall into older BTs. This leads to fewer BTs on irregular, compared to regular datasets.

\subsection{Insertion Performance} \label{subsec:expInsert}

Figure \ref{fig:InsertFull} shows the I/O cost (total number of page read and write operations) after inserting all records of TAXI, BIKE, NFT and AMAZON, using a 4MiB LRU buffer. TIDE, SEB and CEB follow similar append-only frameworks, caching the last (mutable) data node of their most recent BTs, whereas the RI-tree caches disk pages of recent buckets. The number above each bar indicates how many times each method is more costly than TIDE. For regular datasets, TIDE is a few times faster than the rest, while for the irregular ones, it is hundreds of times faster. 

\begin{figure}[t]
	\centering
        \includegraphics[width=\linewidth]{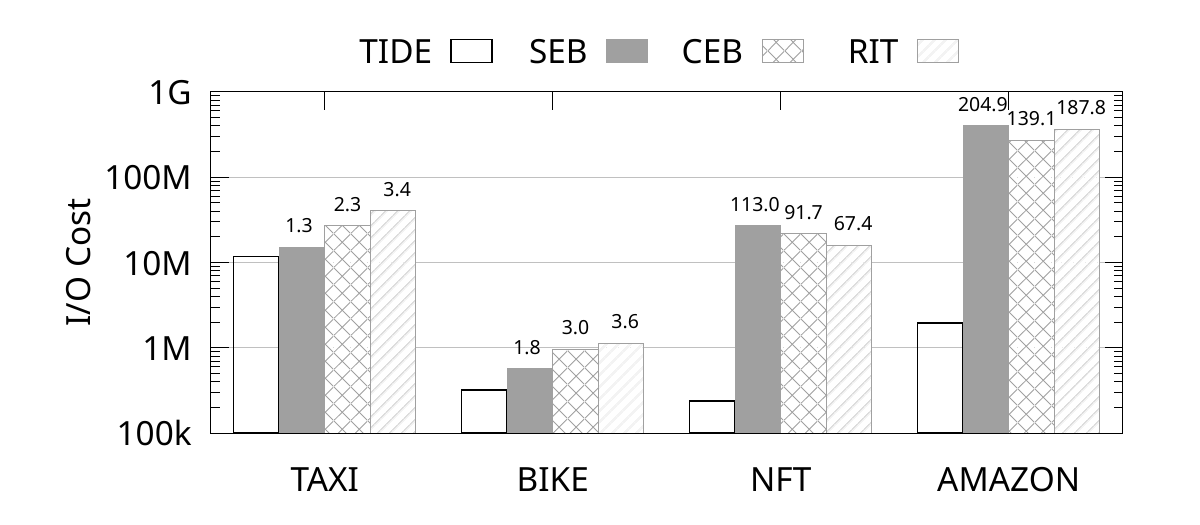}
	\caption{I/O cost of inserting the full dataset}
	\label{fig:InsertFull}
\end{figure}

Given the low number of bottom trees (see Figure \ref{fig:data-nodes}), TIDE  keeps in the LRU buffer the last (mutable) data node of every BT, and the entire top tree. This is not possible for SEB and CEB because of their numerous BTs. Thus, searching for the proper data node to accommodate an insertion incurs I/O cost. Moreover, SEB and CEB suffer from numerous vertical splits (on a large top tree).
For the RI-tree, insertions in TAXI and BIKE often fall into recent buckets (e.g., $b_{28}$ in Figure \ref{fig:RITCorner}), which are cached. However, it still incurs three times more I/O cost than TIDE due to its larger index size.
For the irregular datasets (NFT and AMAZON), high duration variance impacts cache locality in SEB (CEB) based on $t_s$ ($c$), i.e., the insertion of long intervals necessitates fetching from the disk nodes with low $t_s$ ($c$), which have not been accessed recently. CEB saves I/O cost over SEB because its insertions must fall into mutable BTs (i.e., $c \geq \frac{now}{2}$), enabling better cache locality. In comparison to SEB, the RI-tree demonstrates better cache locality on irregular datasets because its insertions fall into very few buckets (e.g., $b_{16}$, $b_{24}$ and $b_{28}$ in Figure \ref{fig:RITCorner}). 

\begin{figure*}[b]
	\centering
        \includegraphics[width=\linewidth]{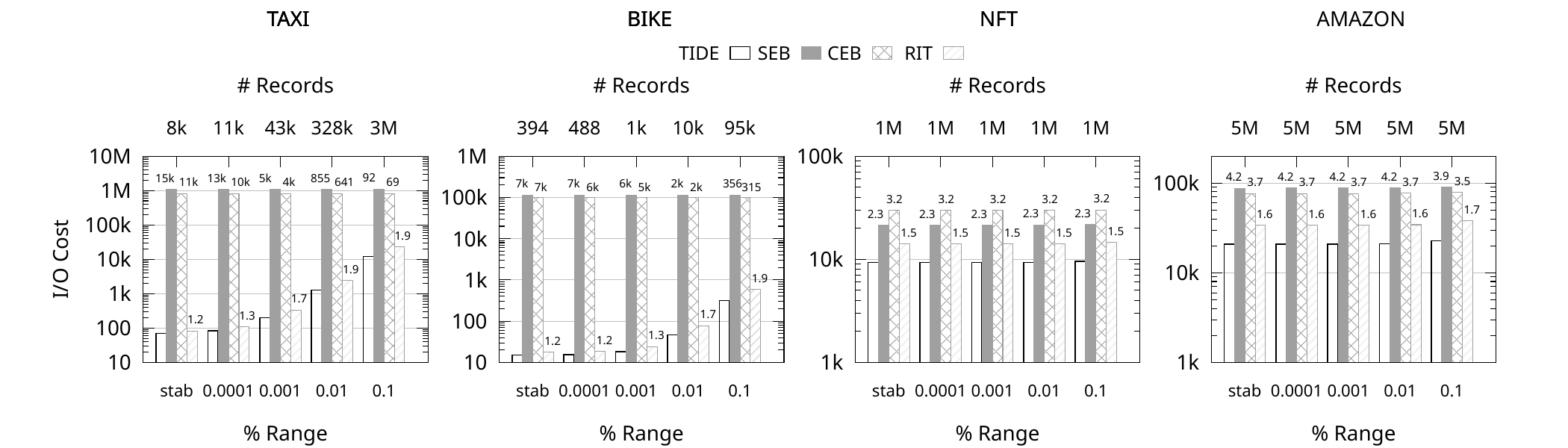}
	\caption{Range queries}
	\label{fig:QueryRange}
\end{figure*}

\subsection{Query Performance} \label{subsec:expQuery}
Figure \ref{fig:QueryRange} shows the I/O cost of stabbing queries and ranges covering $0.0001\%$ to $0.1\%$ of the total history. Each reported result is the average of 1000 uniformly distributed queries. The mean output cardinality is above the diagrams.
The number on top of each bar shows the ratio of the query cost over TIDE. TIDE is up to several orders of magnitude faster than SEB and CEB on regular datasets, and a few times faster on irregular ones.  CEB is better than SEB in datasets (TAXI, BIKE, AMAZON) that achieves high compression (22\% - 30\%). The exception is NFT, where the compression is only 8\%. Naturally, the relative performance difference of SEB and CEB (wrt to TIDE) decreases with the output cardinality \footnote{Observe that even a stabbing query retrieves 1 million (5 million) intervals in NFT (AMAZON), and this number remains almost the same for all ranges. This is due to the long average interval duration in irregular datasets (see Table \ref{table:datasets}). } since all methods have to access the data nodes with qualifying intervals. Compared to the RI-tree, TIDE is between 20\% to 90\% cheaper, independently of the dataset characteristics.

The main reason for the superiority of TIDE is a low number of BTs, whose high level nodes reside in the LRU buffer, reducing disk accesses.
In contrast, SEB searches numerous BTs with $t_s \leq q_s$ because any interval starting before $q_s$ may die after $q_e$.
For instance, a stabbing query in the middle of the data space is expected to visit at least half of the 1072630 BTs on TAXI, which cannot be cached, leading to frequent disk accesses. Moreover, SEB has another serious weakness: BTs with small $t_s$ rarely contain results, although their mutable nodes intersect $q_s$. Such nodes can be very long (they extend to the current time), but are not likely to receive insertions because most regular intervals are short. 
CEB has a similar weakness: BTs with large centers ($c \in [q_e, \frac{q_e+now}{2}]$) rarely contain results, but require accesses. The absolute cost of SEB and CEB remains rather stable with the output cardinality, indicating that it is dominated by visits to irrelevant nodes. 
Despite a large index size, the RI-tree has competitive query performance because very few ($\leq 2\cdot h$) accessed data nodes contain irrelevant intervals.

\begin{figure*}[t]
	\centering
        \includegraphics[width=\linewidth]{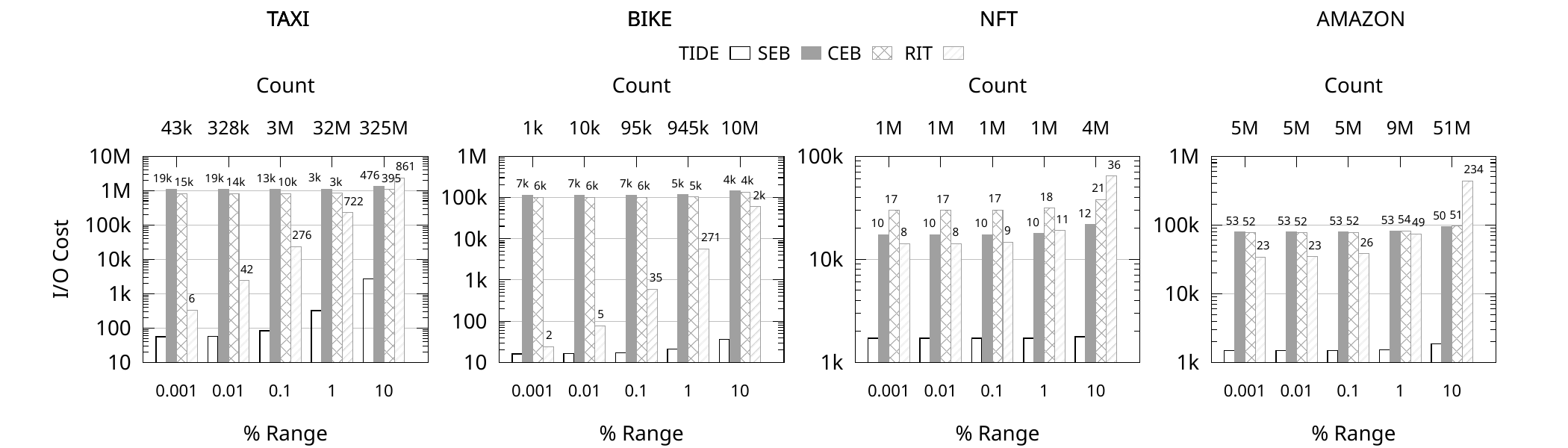}
	\caption{Count queries}
	\label{fig:QueryAgg}
\end{figure*}

Figure \ref{fig:QueryAgg} shows the I/O cost of \emph{count} queries for ranges covering $0.001\%$ to $10\%$ of history.
A count query only returns the number of qualifying records, shown on top of each plot, instead of retrieving their IDs.
Compared to conventional ranges, they incur less I/O cost in TIDE, CEB and SEB because they aggregate directly the results of full nodes covered by the query. Only partially intersecting and mutable nodes need to be visited.
The benefits of TIDE are even more substantial in this setting, outperforming CEB and SEB 3-5 (1-2) orders of magnitude in regular (irregular) datasets. This is because CEB and SEB examine numerous irrelevant mutable nodes, which cannot be covered since they are open-ended. The RI-tree must access all data nodes because they are all mutable. Consequently, its cost is the same as that of conventional range queries, and its relative performance to TIDE is much worse than Figure \ref{fig:QueryRange}.

\section{Conclusion} \label{sec:conclude}
Under the IET assumption, intervals are inserted in increasing order of their ending time $t_e$. SEB, the only existing index aimed at IET, follows a 2-level architecture where a top tree orders intervals on starting time $t_s$, while the bottom trees organize intervals on $t_e$. We first propose CEB, as an optimization of SEB, which reduces the index size.
The top tree of CEB organizes intervals by their center point $c$; accordingly, bottom trees with $c<\frac{now}{2}$ are immutable and compressed. 
Nevertheless, CEB suffers from similar drawbacks as SEB, namely, due to the ever increasing value of $c$ and $t_s$, the number of BTs is unbounded. To avoid this problem, our second contribution TIDE organizes intervals at the top tree based on their duration $d$. This leads to a very small number of bottom trees, compact index size, good cache locality and improved update/query performance. 

The IET assumption facilitates fast insertions through append-only B+trees. In addition, it enables immutable nodes that can be full, minimizing the total index size. Subsequently, all indexes aimed at IET have smaller size and faster insertions than other disk-resident interval indexes. For instance, SEB, CEB and TIDE are about three times smaller than the state-of-the-art RI-tree for all the evaluated datasets and faster for insertions of regular datasets. Moreover, TIDE is always the most efficient on query processing under all settings. Compared to the RI-tree, whose 2-tree structure is highly optimized for range queries, TIDE is also more flexible and can effectively process a variety of other tasks, including count and duration queries.    

\begin{acknowledgements}
	This work was supported by GRF grant 16208623 from Hong Kong RGC.
\end{acknowledgements}

\bibliographystyle{template/spmpsci}
\bibliography{ref}

@inproceedings{Sellis1987RPlus,
author = {Sellis, Timos K. and Roussopoulos, Nick and Faloutsos, Christos},
title = {The R+-Tree: A Dynamic Index for Multi-Dimensional Objects},
year = {1987},
isbn = {093461346X},
publisher = {Morgan Kaufmann Publishers Inc.},
address = {San Francisco, CA, USA},
booktitle = {Proceedings of the 13th International Conference on Very Large Data Bases},
pages = {507–518},
numpages = {12},
series = {VLDB '87}
}

@article{Beckmann1990RStar,
author = {Beckmann, Norbert and Kriegel, Hans-Peter and Schneider, Ralf and Seeger, Bernhard},
title = {The R*-tree: an efficient and robust access method for points and rectangles},
year = {1990},
issue_date = {Jun. 1990},
publisher = {Association for Computing Machinery},
address = {New York, NY, USA},
volume = {19},
number = {2},
issn = {0163-5808},
url = {https://doi.org/10.1145/93605.98741},
doi = {10.1145/93605.98741},
journal = {SIGMOD Rec.},
month = may,
pages = {322–331},
numpages = {10}
}

@incollection{Nievergelt1987midpoint,
  title={Storage and access structures for geometric data bases},
  author={Nievergelt, J{\"u}rg and Hinrichs, Klaus},
  booktitle={Foundations of Data Organization},
  pages={441--455},
  year={1987},
  publisher={Springer}
}

@inproceedings{Seeger1988midpoint,
author = {Seeger, Bernhard and Kriegel, Hans-Peter},
title = {Techniques for Design and Implementation of Efficient Spatial Access Methods},
year = {1988},
isbn = {0934613753},
publisher = {Morgan Kaufmann Publishers Inc.},
address = {San Francisco, CA, USA},
booktitle = {Proceedings of the 14th International Conference on Very Large Data Bases},
pages = {360–371},
numpages = {12},
series = {VLDB '88}
}

@article{Gaede1998survey,
author = {Gaede, Volker and G\"{u}nther, Oliver},
title = {Multidimensional access methods},
year = {1998},
issue_date = {June 1998},
publisher = {Association for Computing Machinery},
address = {New York, NY, USA},
volume = {30},
number = {2},
issn = {0360-0300},
url = {https://doi.org/10.1145/280277.280279},
doi = {10.1145/280277.280279},
journal = {ACM Comput. Surv.},
month = jun,
pages = {170–231},
numpages = {62}
}

@inproceedings{Faloutsos1991DOT,
author = {Faloutsos, Christos and Rong, Yi},
title = {DOT: A Spatial Access Method Using Fractals},
year = {1991},
isbn = {0818621389},
publisher = {IEEE Computer Society},
address = {USA},
booktitle = {Proceedings of the Seventh International Conference on Data Engineering},
pages = {152–159},
numpages = {8}
}

@inproceedings{Kanellakis1993corner,
author = {Kanellakis, Paris C. and Ramaswamy, Sridhar and Vengroff, Darren E. and Vitter, Jeffrey S.},
title = {Indexing for data models with constraints and classes (extended abstract)},
year = {1993},
isbn = {0897915933},
publisher = {Association for Computing Machinery},
address = {New York, NY, USA},
url = {https://doi.org/10.1145/153850.153884},
doi = {10.1145/153850.153884},
booktitle = {Proceedings of the Twelfth ACM SIGACT-SIGMOD-SIGART Symposium on Principles of Database Systems},
pages = {233–243},
numpages = {11},
location = {Washington, D.C., USA},
series = {PODS '93}
}

@article{Goh1996IST,
author = {Goh, Cheng Hian and Lu, Hongjun and Ooi, Beng-Chin and Tan, Kian-Lee},
title = {Indexing temporal data using existing B+-trees},
year = {1996},
issue_date = {March 1996},
publisher = {Elsevier Science Publishers B. V.},
address = {NLD},
volume = {18},
number = {2},
issn = {0169-023X},
url = {https://doi.org/10.1016/0169-023X(95)00034-P},
doi = {10.1016/0169-023X(95)00034-P},
journal = {Data Knowl. Eng.},
month = mar,
pages = {147–165},
numpages = {19}
}

@INPROCEEDINGS{Arge1996external,
  author={Arge, L. and Vitter, J.S.},
  booktitle={Proceedings of 37th Conference on Foundations of Computer Science}, 
  title={Optimal dynamic interval management in external memory}, 
  year={1996},
  volume={},
  number={},
  publisher = {IEEE Comput. Soc. Press},
  pages={560-569},
  doi={10.1109/SFCS.1996.548515}
}

@article{Vitter2001external,
author = {Vitter, Jeffrey Scott},
title = {External memory algorithms and data structures: dealing with massive data},
year = {2001},
issue_date = {June 2001},
publisher = {Association for Computing Machinery},
address = {New York, NY, USA},
volume = {33},
number = {2},
issn = {0360-0300},
url = {https://doi.org/10.1145/384192.384193},
doi = {10.1145/384192.384193},
journal = {ACM Comput. Surv.},
month = jun,
pages = {209–271},
numpages = {63}
}

@inproceedings{Song2003SEB,
author = {Song, Zhexuan and Roussopoulos, Nick},
title = {SEB-tree: An Approach to Index Continuously Moving Objects},
year = {2003},
isbn = {3540003932},
publisher = {Springer-Verlag},
address = {Berlin, Heidelberg},
booktitle = {Proceedings of the 4th International Conference on Mobile Data Management},
pages = {340–344},
numpages = {5},
series = {MDM '03}
}

@inproceedings{Wang2008CSE,
author = {Wang, Longhao and Zheng, Yu and Xie, Xing and Ma, Wei-Ying},
title = {A Flexible Spatio-Temporal Indexing Scheme for Large-Scale GPS Track Retrieval},
year = {2008},
isbn = {9780769531540},
publisher = {IEEE Computer Society},
address = {USA},
url = {https://doi.org/10.1109/MDM.2008.24},
doi = {10.1109/MDM.2008.24},
booktitle = {Proceedings of the The Ninth International Conference on Mobile Data Management},
pages = {1–8},
numpages = {8},
series = {MDM '08}
}

@inproceedings{Mamoulis2010CornerR,
author = {U, Leong Hou and Mamoulis, Nikos and Berberich, Klaus and Bedathur, Srikanta},
title = {Durable top-k search in document archives},
year = {2010},
isbn = {9781450300322},
publisher = {Association for Computing Machinery},
address = {New York, NY, USA},
url = {https://doi.org/10.1145/1807167.1807228},
doi = {10.1145/1807167.1807228},
booktitle = {Proceedings of the 2010 ACM SIGMOD International Conference on Management of Data},
pages = {555–566},
numpages = {12},
location = {Indianapolis, Indiana, USA},
series = {SIGMOD '10}
}

@article{Lu2019TDSQL,
author = {Lu, Wei and Zhao, Zhanhao and Wang, Xiaoyu and Li, Haixiang and Zhang, Zhenmiao and Shui, Zhiyu and Ye, Sheng and Pan, Anqun and Du, Xiaoyong},
title = {A lightweight and efficient temporal database management system in TDSQL},
year = {2019},
issue_date = {August 2019},
publisher = {VLDB Endowment},
volume = {12},
number = {12},
issn = {2150-8097},
url = {https://doi.org/10.14778/3352063.3352122},
doi = {10.14778/3352063.3352122},
journal = {Proc. VLDB Endow.},
month = aug,
pages = {2035–2046},
numpages = {12}
}

@article{Edelsbrunner1980ItvTree,
  title={Dynamic rectangle intersection searching},
  author={Edelsbrunner, H},
  journal={Technical Report},
  pages={47},
  year={1980}
}

@book{Berg2000SegTree,
author = {Berg, Mark de and Cheong, Otfried and Kreveld, Marc van and Overmars, Mark},
title = {Computational Geometry: Algorithms and Applications},
year = {2008},
isbn = {3540779736},
publisher = {Springer-Verlag TELOS},
address = {Santa Clara, CA, USA},
edition = {3rd ed.}
}

@inproceedings{Elmasri1990TimeIdx,
author = {Elmasri, Ramez and Wuu, Gene T. J. and Kim, Yeong-Joon},
title = {The time index—an access structure for temporal data},
year = {1990},
isbn = {055860149X},
publisher = {Morgan Kaufmann Publishers Inc.},
address = {San Francisco, CA, USA},
booktitle = {Proceedings of the Sixteenth International Conference on Very Large Databases},
pages = {1–12},
numpages = {12},
location = {Brisbane, Australia}
}

@article{Kolovson1991SegR,
author = {Kolovson, Curtis P. and Stonebraker, Michael},
title = {Segment indexes: dynamic indexing techniques for multi-dimensional interval data},
year = {1991},
issue_date = {June 1991},
publisher = {Association for Computing Machinery},
address = {New York, NY, USA},
volume = {20},
number = {2},
issn = {0163-5808},
url = {https://doi.org/10.1145/119995.115807},
doi = {10.1145/119995.115807},
journal = {SIGMOD Rec.},
month = apr,
pages = {138–147},
numpages = {10}
}

@article{Segev1993APtree,
author = {Segev, A. and Gunadhi, H.},
title = {Efficient Indexing Methods for Temporal Relations},
year = {1993},
issue_date = {June 1993},
publisher = {IEEE Educational Activities Department},
address = {USA},
volume = {5},
number = {3},
issn = {1041-4347},
url = {https://doi.org/10.1109/69.224200},
doi = {10.1109/69.224200},
journal = {IEEE Trans. on Knowl. and Data Eng.},
month = jun,
pages = {496–509},
numpages = {14}
}

@article{Salzberg1999survey,
author = {Salzberg, Betty and Tsotras, Vassilis J.},
title = {Comparison of access methods for time-evolving data},
year = {1999},
issue_date = {June 1999},
publisher = {Association for Computing Machinery},
address = {New York, NY, USA},
volume = {31},
number = {2},
issn = {0360-0300},
url = {https://doi.org/10.1145/319806.319816},
doi = {10.1145/319806.319816},
journal = {ACM Comput. Surv.},
month = {jun},
pages = {158–221},
numpages = {64}
}

@inproceedings{Kriegel2000RItree,
author = {Kriegel, Hans-Peter and P\"{o}tke, Marco and Seidl, Thomas},
title = {Managing Intervals Efficiently in Object-Relational Databases},
year = {2000},
isbn = {1558607153},
publisher = {Morgan Kaufmann Publishers Inc.},
address = {San Francisco, CA, USA},
booktitle = {Proceedings of the 26th International Conference on Very Large Data Bases},
pages = {407–418},
numpages = {12},
series = {VLDB '00}
}

@inproceedings{Behrend2019PeriodIdx,
author = {Behrend, Andreas and Dign\"{o}s, Anton and Gamper, Johann and Schmiegelt, Philip and Voigt, Hannes and Rottmann, Matthias and Kahl, Karsten},
title = {Period Index: A Learned 2D Hash Index for Range and Duration Queries},
year = {2019},
isbn = {9781450362801},
publisher = {Association for Computing Machinery},
address = {New York, NY, USA},
url = {https://doi.org/10.1145/3340964.3340965},
doi = {10.1145/3340964.3340965},
booktitle = {Proceedings of the 16th International Symposium on Spatial and Temporal Databases},
pages = {100–109},
numpages = {10},
location = {Vienna, Austria},
series = {SSTD '19}
}

@inproceedings{Ceccarello2023RDtree,
author = {Ceccarello, Matteo and Dign\"{o}s, Anton and Gamper, Johann and Khnaisser, Christina},
title = {Indexing Temporal Relations for Range-Duration Queries},
year = {2023},
isbn = {9798400707469},
publisher = {Association for Computing Machinery},
address = {New York, NY, USA},
url = {https://doi.org/10.1145/3603719.3603732},
doi = {10.1145/3603719.3603732},
booktitle = {Proceedings of the 35th International Conference on Scientific and Statistical Database Management},
articleno = {3},
numpages = {12},
location = {Los Angeles, CA, USA},
series = {SSDBM '23}
}

@inproceedings{George2022HINT,
author = {Christodoulou, George and Bouros, Panagiotis and Mamoulis, Nikos},
title = {HINT: A Hierarchical Index for Intervals in Main Memory},
year = {2022},
isbn = {9781450392495},
publisher = {Association for Computing Machinery},
address = {New York, NY, USA},
url = {https://doi.org/10.1145/3514221.3517873},
doi = {10.1145/3514221.3517873},
booktitle = {Proceedings of the 2022 International Conference on Management of Data},
pages = {1257–1270},
numpages = {14},
location = {Philadelphia, PA, USA},
series = {SIGMOD '22}
}

@article{George2024LIT,
author = {Christodoulou, George and Bouros, Panagiotis and Mamoulis, Nikos},
title = {LIT: Lightning-fast In-memory Temporal Indexing},
year = {2024},
issue_date = {February 2024},
publisher = {Association for Computing Machinery},
address = {New York, NY, USA},
volume = {2},
number = {1},
url = {https://doi.org/10.1145/3639275},
doi = {10.1145/3639275},
journal = {Proc. ACM Manag. Data},
month = {mar},
articleno = {20},
numpages = {27}
}

@inproceedings{Hu2022TempJ,
author = {Hu, Xiao and Sintos, Stavros and Gao, Junyang and Agarwal, Pankaj K. and Yang, Jun},
title = {Computing Complex Temporal Join Queries Efficiently},
year = {2022},
isbn = {9781450392495},
publisher = {Association for Computing Machinery},
address = {New York, NY, USA},
url = {https://doi.org/10.1145/3514221.3517893},
doi = {10.1145/3514221.3517893},
booktitle = {Proceedings of the 2022 International Conference on Management of Data},
pages = {2076–2090},
numpages = {15},
location = {Philadelphia, PA, USA},
series = {SIGMOD '22}
}

@inproceedings{Eltabakh2006SPGiST,
  title = {Space-Partitioning Trees in PostgreSQL: Realization and Performance},
  url = {http://dx.doi.org/10.1109/ICDE.2006.146},
  DOI = {10.1109/icde.2006.146},
  booktitle = {22nd International Conference on Data Engineering (ICDE’06)},
  publisher = {IEEE},
  author = {Eltabakh,  M.Y. and Eltarras,  R. and Aref,  W.G.},
  year = {2006},
  pages = {100–100}
}

@article{Bouros2021optFS,
author = {Bouros, Panagiotis and Mamoulis, Nikos and Tsitsigkos, Dimitrios and Terrovitis, Manolis},
title = {In-Memory Interval Joins},
year = {2021},
issue_date = {Jul 2021},
publisher = {Springer-Verlag},
address = {Berlin, Heidelberg},
volume = {30},
number = {4},
issn = {1066-8888},
url = {https://doi.org/10.1007/s00778-020-00639-0},
doi = {10.1007/s00778-020-00639-0},
journal = {The VLDB Journal},
month = apr,
pages = {667–691},
numpages = {25}
}

@article{Bouros2025HINTJoin,
  title={Querying Interval Data on Steroids},
  author={Bouros, Panagiotis and Christodoulou, George and Rauch, Christian and Titkov, Artur and Mamoulis, Nikos},
  journal={IEEE Transactions on Knowledge and Data Engineering},
  year={2025},
  publisher={IEEE}
}

@article{Costa2023nft,
  title={Unraveling the NFT economy: A comprehensive collection of Non-Fungible Token transactions and metadata},
  author={Costa, Davide and La Cava, Lucio and Tagarelli, Andrea},
  journal={Data in Brief},
  volume={51},
  pages={109749},
  year={2023},
  publisher={Elsevier}
}

@article{Hou2024amazon,
  title={Bridging Language and Items for Retrieval and Recommendation},
  author={Hou, Yupeng and Li, Jiacheng and He, Zhankui and Yan, An and Chen, Xiusi and McAuley, Julian},
  journal={arXiv preprint arXiv:2403.03952},
  year={2024}
}

@inproceedings{Driscoll1986Persis,
author = {Driscoll, J R and Sarnak, N and Sleator, D D and Tarjan, R E},
title = {Making data structures persistent},
year = {1986},
isbn = {0897911938},
publisher = {Association for Computing Machinery},
address = {New York, NY, USA},
url = {https://doi.org/10.1145/12130.12142},
doi = {10.1145/12130.12142},
booktitle = {Proceedings of the Eighteenth Annual ACM Symposium on Theory of Computing},
pages = {109–121},
numpages = {13},
location = {Berkeley, California, USA},
series = {STOC '86}
}

@article{Lanka1991FullyPersis,
  title={Fully persistent B+-trees},
  author={Lanka, Sitaram and Mays, Eric},
  journal={ACM SIGMOD Record},
  volume={20},
  number={2},
  pages={426--435},
  year={1991}
}

@inproceedings{Brodal2023FullyPersis,
author = {Brodal, Gerth St\o{}lting and Rysgaard, Casper Moldrup and Svenning, Rolf},
title = {External Memory Fully Persistent Search Trees},
year = {2023},
isbn = {9781450399135},
publisher = {Association for Computing Machinery},
address = {New York, NY, USA},
url = {https://doi.org/10.1145/3564246.3585140},
doi = {10.1145/3564246.3585140},
booktitle = {Proceedings of the 55th Annual ACM Symposium on Theory of Computing},
pages = {1410–1423},
numpages = {14},
location = {Orlando, FL, USA},
series = {STOC 2023}
}

@inproceedings{Cudre2010Trajstore,
  title={Trajstore: An adaptive storage system for very large trajectory data sets},
  author={Cudre-Mauroux, Philippe and Wu, Eugene and Madden, Samuel},
  booktitle={2010 IEEE 26th International Conference on Data Engineering (ICDE 2010)},
  pages={109--120},
  year={2010},
  organization={IEEE}
}

@inproceedings{Patel2004STRIPES,
  title={STRIPES: an efficient index for predicted trajectories},
  author={Patel, Jignesh M and Chen, Yun and Chakka, V Prasad},
  booktitle={Proceedings of the 2004 ACM SIGMOD international conference on Management of data},
  pages={635--646},
  year={2004}
}

@article{De2005MonTree,
  title={Indexing the trajectories of moving objects in networks},
  author={De Almeida, Victor Teixeira and G{\"u}ting, Ralf Hartmut},
  journal={GeoInformatica},
  volume={9},
  number={1},
  pages={33--60},
  year={2005},
  publisher={Springer}
}

@INPROCEEDINGS{Amagata2024Sampling,
  author={Amagata, Daichi},
  booktitle={2024 IEEE 40th International Conference on Data Engineering (ICDE)}, 
  title={Independent Range Sampling on Interval Data}, 
  year={2024},
  volume={},
  number={},
  pages={449-461},
  doi={10.1109/ICDE60146.2024.00041}
}
\end{document}